%% file: PCforTensorIso.tex
\documentclass{article}
\usepackage{amsmath, amssymb, amsthm}
\usepackage[left=1in,right=1in,top=1in,bottom=1in]{geometry}
\usepackage{hyperref}
\usepackage{arydshln} 

\newtheorem{theorem}{Theorem}[section]

\newtheorem{lemma}[theorem]{Lemma}
\newtheorem{proposition}[theorem]{Proposition}
\newtheorem{observation}[theorem]{Observation}

\newtheorem{corollary}[theorem]{Corollary}

\theoremstyle{definition}
\newtheorem{definition}[theorem]{Definition}
\newtheorem{conjecture}[theorem]{Conjecture}

\newtheorem{remark}[theorem]{Remark}
\newtheorem{question}[theorem]{Open Question}

\DeclareMathOperator{\GL}{GL}
\DeclareMathOperator{\Span}{Span}
\DeclareMathOperator{\Id}{Id}

\DeclareMathOperator{\rk}{rk}
\newcommand{\F}{\mathbb{F}}
\newcommand{\R}{\mathbb{R}}

\newcommand{\diag}{\text{diag}}

\title{On the algebraic proof complexity of Tensor Isomorphism}
\author{Nicola Galesi\footnote{Dipartimento Ingegneria Informatica Automatica e Gestionale ``A. Ruberti'', Sapienza Università di Roma, Italy, \texttt{nicola.galesi@uniroma1.it}}, Joshua A. Grochow\footnote{Departments of Computer Science and Mathematics, University of Colorado Boulder, \texttt{jgrochow@colorado.edu}. Supported by NSF CAREER award CCF-2047756.}, Toniann Pitassi\footnote{Department of Computer Science, Columbia University, \texttt{tonipitassi@gmail.com}. Supported by NSF grant CCF-1900460, and by the IAS School of Mathematics.}, and Adrian She\footnote{Department of Mathematics and Computer Science, University of Toronto, \texttt{adrian.she@mail.utoronto.ca}. Supported by NSERC Canada Graduate Scholarship.}}

\begin{document}
\maketitle

\input{intro.tex}

\input{prelim.tex}

\input{matrices.tex}

\input{tensors.tex}

\section{Open Questions}
Beyond Conjecture~\ref{conj:inv}, we highlight several more questions we find interesting about the algebraic proof complexity of \textsc{Tensor Isomorphism}.

\subsection{Degree}
\begin{question}
What is the correct value for the PC degree of rank-$r$ \textsc{Tensor Isomorphism}?
\end{question}

Note that by using the reductions from Section~\ref{sec:lowerbound}, we can produce (random) $r \times r \times r$ tensors that require PC degree $\Omega(r^{1/4})$ to refute. However, the number of variables is $6r^2$, this lower bound is only $\Omega(N^{1/8})$ where $N$ is the number of variables. Since their rank could be as large as $R = \Theta(r^2)$ (and indeed, very likely is), the upper bound we get from Theorem~\ref{thm:lowrank} is only $2^{O(r^4)}$ (without the $x^q - x$ axioms) or $O(r^4)$ (with the $x^q - x$ axioms, with $q = O(1)$). Even in the latter case, this leaves a polynomial gap between the lower and upper bounds (without those the gap is exponential). 

We note that the upper bound in Theorem~\ref{thm:lowrank} without the $x^q-x$ equations already applies to the weaker Nullstellensatz proof system. Is there a polynomial upper bound on PC degree---as a function of rank---without the $x^q-x$ axioms?

\subsection{Size}
In the presence of the Boolean axioms, there is a size-degree tradeoff for PC (or even PCR---a system with the same degree bounds as PC, but is stronger when measuring size by number of monomials or number of symbols) \cite{CEI, ABRW}. This implies that in the presence of the Boolean axioms, a good degree lower bound implies a good size lower bound. But \textsc{TI} does not have the Boolean axioms. 

\begin{question}
Get lower and upper bounds on the \emph{size} of PC proofs for \textsc{Tensor (Non-)Isomorphism}. Are there subexponential size upper bounds, despite the polynomial degree lower bounds?
\end{question}

\subsection{Other matrix problems}
While many different tensor-related problems are all equivalent to \textsc{TI}, in the case of matrices, we have three genuinely different problems: matrix equivalence (2-TI), matrix conjugacy, and matrix congruence. Conjugacy is determined by the Rational Normal Form or Jordan Normal Form, while congruence depends on the field (e.g., over algebraically closed fields it only depends on rank, over $\mathbb{R}$ it depends on the signature, and over finite fields it depends on whether the determinant is a square or not).

\begin{question}
What is the PC complexity (size, degree, etc.) of matrix conjugacy? Of matrix congruence?
\end{question}

More precisely, for conjugacy we have in mind the system of equations:
\[
XM = M'X \qquad XX' = X'X = I,
\]
and for congruence the system of equations:
\[
XMX^T = M' \qquad XX' = X'X = I.
\]
\subsection{Bounded border rank}
Not only can testing a tensor for bounded rank can be done in polynomial time (Remark~\ref{rmk:rank}), testing a tensor for bounded \emph{border}-rank can also be done in polynomial time (see, e.\,g., \cite{GrochowSE}), by evaluating a polynomial number of easy-to-evaluate equations. While several partial results are available, the gap for what is known about the ratio between rank and border rank is quite large: there are 3-tensors known whose rank approaches 3 times their border rank \cite{zuiddam}, but the currently known upper bound is Lehmkuhl and Lickteig \cite{LL}, who show that for tensors of border rank $b$, the ratio of rank to border rank is at most $c^{\Theta(nb)}$. See the Zuiddam's introduction \cite{zuiddam} for more details.

\begin{question}
What is the PC degree of testing isomorphism of tensors of bounded border-rank? Can such tests be done (by any method) in polynomial time?
\end{question}

\subsection{Relating different reductions from \textsc{Graph Isomorphism}}
While we chose a particular reduction from \textsc{GI} to \textsc{TI} for the lower bound in Section~\ref{sec:lowerboundGI}, we are aware of several others, including: 
\begin{itemize}
\item \textsc{GI} to \textsc{Permutational Code Equivalence} \cite{PR, Luks, Miyazaki}, then to \textsc{Matrix Lie Algebra Conjugacy} \cite{GrochowLie}, then to \textsc{TI} \cite{FGS};
 
\item \textsc{GI} to \textsc{Semisimple Matrix Lie Algebra Conjugacy} \cite{GrochowLie}, and then to \textsc{TI} \cite{FGS}; 

\item \textsc{GI} to \textsc{Alternating Matrix Space Isometry} \cite{GQ, HQ}, then to \textsc{TI} \cite{FGS}; 

\item \textsc{GI} to \textsc{Algebra Isomorphism} \cite{grigorievGI, AS05}, then to \textsc{TI} \cite{FGS}.
\end{itemize}
We believe all of these can be realized as low-degree PC reduction as well. In the first arXiv version of \cite{GQ}, they asked which of these might be equivalent in some sense (though there the final target was \textsc{Alternating Matrix Space Isometry}, another $\mathsf{TI}$-complete problem, rather than \textsc{TI} itself). Here we make this question slightly more precise, in terms of PC reductions:

\begin{question}
Which, if any, of the reductions above from \textsc{Graph Isomorphism} to \textsc{Tensor Isomorphism} are equivalent under low-degree PC?
\end{question}

\bibliographystyle{alphaurl}
\bibliography{tensor-bib}

\end{document}

%% file: intro.tex

\begin{abstract}
The \textsc{Tensor Isomorphism} problem (\textsc{TI}) has recently emerged as having connections to multiple areas of research within complexity and beyond, but the current best upper bound is essentially the brute force algorithm. Being an algebraic problem, \textsc{TI} (or rather, proving that two tensors are \emph{non}-isomorphic) lends itself very naturally to algebraic and semi-algebraic proof systems, such as the Polynomial Calculus (PC) and Sum of Squares (SoS). For its combinatorial cousin \textsc{Graph Isomorphism}, essentially optimal lower bounds are known for approaches based on PC and SoS (Berkholz \& Grohe, SODA '17). Our main results are an $\Omega(n)$ lower bound on PC degree or SoS degree for \textsc{Tensor Isomorphism}, and a nontrivial upper bound for testing isomorphism of tensors of bounded rank.

We also show that PC cannot perform basic linear algebra in sub-linear degree, such as comparing the rank of two matrices, or deriving $BA=I$ from $AB=I$. As linear algebra is a key tool for understanding tensors, we introduce a strictly stronger proof system, PC+Inv, which allows as derivation rules all substitution instances of the implication $AB=I \rightarrow BA=I$. We conjecture that even PC+Inv cannot solve  \textsc{TI} in polynomial time either, but leave open getting lower bounds on PC+Inv for any system of equations, let alone those for \textsc{TI}. We also highlight many other open questions about proof complexity approaches to \textsc{TI}.
\end{abstract}

\section{Introduction}
Tensors have rapidly emerged as a fundamental data structure and key mathematical object of the 21st century. They play key roles in many different areas of science, engineering, and mathematics, from quantum mechanics and general relativity to neural networks \cite{TNN} and mechanical engineering. They arise in theoretical computer science in many ways, including from (post-quantum) cryptography \cite{Pat96, JQSY19}, derandomization, \textsc{Matrix Multiplication}, \textsc{Graph Isomorphism} \cite{GQ}, and several different parts of Geometric Complexity Theory. 

The fundamental notion of equivalence between tensors is that of isomorphism: two tensors are isomorphic if one can be transformed into the other by an invertible linear change of basis in each of the corresponding vector spaces. For example, two 2-tensors (=matrices) $M,M'$ are equivalent under this notion if there are invertible matrices $X,Y$ such that $XMY = M'$; similarly, two 3-tensors, represented by 3-way arrays $T_{ijk}, T'_{ijk}$ are isomorphic if there are three invertible matrices $X,Y,Z$ such that 
\begin{equation} \label{eq:TI}
\sum_{ijk} X_{ii'} Y_{jj'} Z_{kk'} T_{ijk} = T'_{i'j'k'}
\end{equation}
for all $i',j',k'$. The problem of (3-)\textsc{Tensor Isomorphism} (\textsc{TI}) is: given two such 3-way arrays, to decide if they are isomorphic. 

Over finite fields, two different versions of \textsc{TI} sandwich the complexity of its more famous cousin, \textsc{Graph Isomorphism}. Namely, as presented above, \textsc{GI} reduces to \textsc{TI}. In the other direction, over a finite field $\F=\F_{p^a}$, one can take an $n \times n \times n$ tensor and list it out ``verbosely'', as a set of $p^{an}$ many $n \times n$ matrices over $\F$; the isomorphism problem for such verbosely given tensors is equivalent to \textsc{Group Isomorphism} for a certain class of $p$-groups, widely believed to be the hardest cases of \textsc{Group Isomorphism} in general. As such, this verbose version of \textsc{TI} reduces to \textsc{GI}. Furthermore, with Babai's quasi-polynomial-time algorithm \cite{Babai}, the running times are quite close: $N^{O(\log N)}$ for \textsc{VerboseTI} and $N^{O(\log^2 N)}$ for \textsc{GI} (the exponent of the exponent was worked out by Helfgott \cite{helfgott}). Thus \textsc{TI} stands as a key obstacle to putting \textsc{GI} into $\mathsf{P}$.

In this paper, we initiate the study of (algebraic) proof complexity approaches to proving that two tensors are non-isomorphic.
Lower bounds on the Polynomial Calculus proof system imply lower bounds on Gr\"{o}bner basis techniques, and the latter are some of the leading methods for solving $\mathsf{TI}$-complete problems in cryptanalysis, e.g., \cite{DBLP:conf/eurocrypt/TangDJPQS22, DBLP:conf/eurocrypt/FaugereP06}.
In the context of \textsc{GI}, proof complexity plays an important role, through its connection with the Weisfeiler--Leman (WL) algorithm. Although this algorithm does not, on its own, solve \textsc{GI} in polynomial time \cite{CFI}, it is a key subroutine in many of the best algorithms for \textsc{GI}, both in theory \cite{Babai} and in practice (see \cite{mckay, MP}). And the picture that has emerged is that some proof systems for \textsc{GI} are known to be equivalent in power to WL \cite{AM, BerkholzG15}, and some lower bounds on proof systems are closely related to lower bounds for WL \cite{SSC, DBLP:conf/soda/ODonnellWWZ14, BerkholzG15}. Versions of WL for groups, and in particular finite $p$-groups---and hence, by the connection above, tensors over finite fields---have only recently begun to be explored \cite{BGLQW, BS1, BS2, CL}.

\subsection{Main results}

We focus on the Polynomial Calculus (PC, or Gr\"{o}bner) proof system \cite{CEI}, though our results will also hold for semi-algebraic proof systems such as Sum-of-Squares \cite{lasserre} as well. PC is used to show that a system of polynomial equations over a field $\F$ is unsatisfiable over the algebraic closure $\overline{\F}$, by deriving from the system of equations, in a line-by-line fashion, the contradiction $1=0$. The degree of a PC proof is the maximum degree of any line appearing in the proof, and it is a fundamental result that PC proofs of constant degree can be found in polynomial time \cite{CEI}. Much as WL informally ``captures all combinatorial approaches'' to \textsc{GI}, PC informally ``captures all approaches based on Gr\"{o}bner bases'' to showing that a system of polynomial equations is unsatisfiable. 

The systems of equations we study are, for two non-isomorphic tensors $T,T'$, the equations (\ref{eq:TI}) along with new matrices $X',Y',Z'$, and equations saying that these are the inverses of $X,Y,Z$, resp., viz.: $XX'=X'X = \Id$, and similarly for the others. The reason for introducing these new matrices, despite their not appearing in (\ref{eq:TI}), is that these invertibility equations are only degree 2. In contrast, if we instead used the determinant to indicate that $X$ was invertible, then our starting equations would have degree $n+1$, rather than constant degree $\leq 3$. Since the main complexity measure we study on PC is degree, having starting equations of degree $n$ would make it difficult to make meaningful lower bound statements.

Our first main result is (two proofs of) a lower bound on such techniques.

\begin{theorem} \label{thm:lowerMain}
Over any field, there are instances of $n \times n \times n$ \textsc{Tensor Isomorphism} that require PC degree $\Omega(n)$ to refute. Over $\R$, they also require Sum-of-Squares degree $\Omega(n)$ to refute.
\end{theorem}

The preceding goes by reduction from known lower bounds on PC for \textsc{Graph Isomorphism} \cite{BerkholzG15, BerkholzG16}, but has the disadvantage (from the tensor point of view) that the resulting tensors are quite sparse: in one direction, one of the slices is supported on an $\Omega(n) \times n$ matrix and all the others slices have support size 1. In a second proof (Section~\ref{sec:lowerbound}), we get a polynomially worse lower bound $\Omega(\sqrt[4]{n})$, but with a reduction from \textsc{Random 3XOR} that is more direct. Indeed, we show that 3XOR itself can be viewed as a particular instance of a tensor problem \emph{without} gadgets; gadgets are only then needed to reduce from that tensor problem to \textsc{Tensor Isomorphism} itself. In contrast, the lower bounds on PC for \textsc{GI} (\emph{ibid.}) already use the Cai--F\"{u}rer--Immerman gadgets \cite{CFI} to reduce from XOR-SAT, and then even further gadgets are needed to reduce from \textsc{GI} to \textsc{TI}. 

Our technical contributions in the above theorem are thus three-fold: 
\begin{enumerate}
\item We show that the known reductions from \textsc{GI} to \textsc{TI} can be carried out in low-degree PC; 
\item We realize 3XOR very naturally as a tensor problem; and 
\item We give new reductions from 3XOR, through a series of tensor-related problems, to \textsc{TI}, that work as many-one reductions of the decision problems that can be carried out in low-degree PC.
\end{enumerate}

Complementing our lower bound, we also show that tensors of low rank are comparatively easy to test for (non)-isomorphism. Here, one of our upper bounds is in the weaker Nullstellensatz proof system (giving a stronger upper bound than only a PC upper bound). In the Nullstellensatz proof system, a proof that a system of equations $f_1 = \dotsb = f_m = 0$ is unsatisfiable consists of polynomials $g_i$ such that $\sum g_i f_i = 1$, and the Nullstellensatz degree is the maximum degree of any $g_i f_i$. The PC degree is always at most the Nullstellensatz degree, and the gap between the two can be nearly maximal for Boolean equations: $O(1)$ versus $\Omega(n / \log n)$ \cite{BCIP}. (For Boolean equations, there is always an $O(n)$ upper bound, though this does not apply to \textsc{TI}, see Remark~\ref{rmk:boolean} below).

\begin{theorem} \label{thm:upper}
Over any field, the Nullstellensatz degree of refuting isomorphism of two $n \times n \times n$ tensors of tensor rank $\leq r$ is at most $2^{O(r^2)}$. If working over a finite field $\F_q$ and including the equations $x^q - x$, the PC degree is at most $O(qr^2)$. 

In particular, isomorphism of constant-rank tensors can be decided in polynomial time.
\end{theorem}

\begin{remark} \label{rmk:boolean}
In many settings in proof complexity, Boolean axioms such as $x_i^2 = x_i$ or $x_i^2 = 1$ are included among the system of equations, and all such unsatisfiable systems of equations can be refuted in degree $O(n)$ ($n=$\# variables). If this were the case here, the above would only be interesting for very small values of $r$. In contrast, the equations for \textsc{TI} do not include any such Boolean axioms, and as such the naive degree upper bound is exponential in the number of variables. For $n \times n \times n$ tensors, this gives an upper bound of $2^{O(n^2)}$ \cite{sombra}, and thus, Theorem~\ref{thm:upper} gives nontrivial upper bounds all the way up to $r \leq n$. (We note that $n \times n \times n$ tensors can have rank up to $\Theta(n^2)$ \cite{lickteigTypical}.) The proof of Theorem~\ref{thm:upper} shows that for rank-$r$ tensors, \textsc{TI} can essentially be reduced to a system of equations in only $O(r^2)$ variables.
\end{remark}

\begin{remark} \label{rmk:rank}
For fixed $r$, testing if an $n \times n \times n$ tensor has rank $\leq r$ can be done in polynomial time, as follows. This will show that the algorithm of Theorem~\ref{thm:upper} genuinely solves the decision problem, and not just a promise problem. Given an $n \times n \times n$ tensor $T$, consider its three $n \times n^2$ flattenings. Use Gaussian elimination to put each such flattening, separately, into reduced row echelon form. If any of these flattenings has rank $> r$, reject. Otherwise, we get from this a list of $3r$ vectors $u_1, \dotsc, u_r, v_1, \dotsc, v_r, w_1, \dotsc, w_r$, such that $T$ lives in the $r \times r \times r$-dimensional space $\Span\{u_1, \dotsc, u_r\} \otimes \Span\{v_1, \dotsc, v_r\} \otimes \Span\{w_1, \dotsc, w_r\}$. Now in this space we can write down the Brent equations \cite{brent} for $T$ to have rank $\leq r$, which will be $r^3$ cubic equations in $3r^2$ variables (Brent's equations \cite[(5.06)]{brent} were specifically for the matrix multiplication tensor, but analogous equations are easily constructed for arbitrary tensors using the same idea). Since $r$ is constant, these equations may be solved in polynomial time (here we assume that we are either working over a finite field, a finite-degree extension of the rationals---see, for example, Grigoriev \cite{grigoriev}---or in the BSS model over an arbitrary field).
\end{remark}

Lastly, one may wonder why we focus on 3-\textsc{Tensor Isomorphism}, and not some of its many related variants. Indeed, just as there are other equivalence notions for matrices---such as conjugacy $XMX^{-1}$ and congruence $XMX^T$---there are many different kinds of multilinear objects that can be represented by multi-way arrays, including tensors, homogeneous polynomials (commutative or noncommutative), alternating matrix spaces, multilinear maps, and so on, each with their own corresponding notion of isomorphism. While these problems are indeed distinct, they are all equivalent under polynomial-time isomorphisms \cite{FGS, GQ}; such problems are called $\mathsf{TI}$-complete. Even isomorphism of $k$-way tensors (for any fixed $k \geq 3$) is equivalent to isomorphism of 3-tensors \cite{GQ}. This partially justifies our focus on \textsc{3-Tensor Isomorphism}. In the course of proving our reductions for the results stated above, we use many of the gadgets from \cite{FGS, GQ}, and show that such uses also often yield proof complexity reductions as well. Because of the variety of gadgets used in our reductions, we believe that many, if not all, of the gadgets from those results would also yield proof complexity reductions, so the proof complexity of all the known $\mathsf{TI}$-complete problems should be polynomially related.

\subsection{Comparison with linear algebra, a new proof system, and a conjecture} \label{sec:conj}
As linear algebra is part of the core toolkit for understanding tensors, it is natural to wonder how linear algebra can help in algebraic proof complexity approaches to \textsc{TI}. We believe that even if it had the ``full power'' of linear algebra at its disposal ``for free,'' PC could still not solve \textsc{TI} efficiently. We begin to make this precise in this section.

Some basic derivations in linear algebra are to relate the ranks of two matrices and to derive $BA=I$ from $AB=I$ (the Inversion Principle, one of the so-called ``hard matrix identities'' \cite{SC}, only recently shown to have short $\mathsf{NC}^2$-Frege proofs \cite{HT}). Soltys \cite{soltys} and Soltys \& Cook \cite{SC} discuss the relationship between these and other standard implications in linear algebra. We show that PC is not strong enough to prove these in low-degree:

\begin{theorem}
The unsatisfiable system of equations $XY=\Id_n$ where $X$ is $n \times r$ and $Y$ is $r \times n$ with $1 \leq r < n$, requires degree $\geq r/2+1$ to refute in PC, over any field.
\end{theorem}
We refer to this system of equations as the Rank Principle, as refuting them amounts to showing that $\rk \Id_n > r$. 

\begin{theorem}
Any PC derivation of $BA = I$ from $AB=I$, where $A,B$ are $n \times n$ matrices with $\{0,1\}$ entries, requires degree $\geq n/2+1$, over any field.
\end{theorem}

We also observe that the Rank Principle can be derived in low degree from the Inversion Principle. 

Although it remains open whether the Inversion Principle is ``complete'' for linear-algebraic reasoning (see \cite{soltys, SC}), we introduce the proof system PC+Inv in an attempt to capture some linear-algebraic reasoning that seems potentially useful for \textsc{TI}. PC+Inv has all the same derivation rules as PC, but in addition, for any square matrices $A,B$ (whose entries may themselves be polynomials---that is, we allow substitution instances), we have the rule
\[
\frac{AB=I}{BA=I}.
\]
where the antecedent represents the set of $n^2$ equations corresponding to $AB=I$, and similarly the consequent denotes the set of $n^2$ equations $BA=I$ (see \ref{sec:prelim:inv} for more details). Degree is still measured in the usual way, but this rule lets us ``cut out'' the high-degree proof that would usually be required to derive $BA=I$ from $AB=I$. We now formalize our intuition that linear algebra should not suffice to solve \textsc{TI} efficiently in the following:

\begin{conjecture} \label{conj:inv}
\textsc{Tensor Isomorphism} for $n \times n \times n$ tensors requires degree $\Omega(n)$ in PC+Inv, over any field.
\end{conjecture}

Despite the conjecture, we do not yet know how to prove lower bounds on PC+Inv for any unsatisfiable system of equations, let alone those coming from \textsc{TI}. Mod $p$ counting principles (for $p$ different from the characteristic of the field) strike us as potentially interesting instances to examine for PC+Inv lower bounds, before tackling a harder problem like \textsc{TI}. In the final section, we highlight many other open questions around the proof complexity of \textsc{TI}.

\subsection{Organization}
In Section~\ref{sec:prelim} we cover preliminaries. In Section~\ref{sec:matrix} we prove the lower bounds on linear algebraic principles just discussed. In Section~\ref{sec:upper} we prove the upper bound for isomorphism of bounded rank tensors (Theorem~\ref{thm:upper}). In Section~\ref{sec:lowerboundGI} we prove Theorem~\ref{thm:lowerMain} by reduction from \textsc{GI}. In Section~\ref{sec:lowerbound} we prove the polynomially related lower bound by direct reduction from \textsc{Random 3XOR}.

%% file: prelim.tex

\section{Preliminaries} \label{sec:prelim}

\subsection{Proof systems}
All our rings are commutative and unital. {\em Polynomial calculus} (PC) is a proof system to prove that a given system of (multivariate) polynomial equations ${\cal P}$ over a field $\mathbb F$ of the form $p=0$, has no solution over the algebraic closure $\overline{\F}$ (i.e. the system is unsolvable). We usually shorten the polynomial equation $p=0$ to just $p$.
The derivation rules of the system are the following one:
$$
\frac{p}{xp} \; \mbox{(multiplication)}, \quad \frac{p\quad q}{ap+bq} \; \mbox{(linear combination)} 
$$
where $x$ is any formal variable, $a,b \in \mathbb{F}$ and $p,q$ are polynomials over $\mathbb F$.  

When refuting Boolean systems of equations it is common to include the Boolean axioms $x_i^2 - x_i$. Because we do \emph{not} always include these (esp. for \textsc{TI}) we are explicit about our use of these, but do not assume they are built into the proof system---that is, if we are assuming them as axioms, we say so.

A {\em PC derivation (or proof)}  of a  polynomial $q$ from a set of polynomials ${\cal P}$ is a sequence of polynomial equations $p_1,\ldots, p_{m}$ ending with the polynomial $q$ (so $p_m$ is $q$) and where each $p_i$, $i \in [m]$, is either an {\em axiom} $p$ for $p \in  {\cal P}$, or is obtained from previous equations in the refutation by multiplication or linear combination.  We denote this by writing  ${\cal P}\vdash q$.  Observe that if $p$  is derivable in  PC and $q$ is a polynomial then, by repeated applications of multiplication and linear combination rules, we can derive
$pq$. We often use this generalization of the multiplication in our proofs without mention.  

A {\em PC refutation}  is just a PC proof  of the polynomial $1$.   The {\em degree} of a PC derivation is the maximal degree of a polynomial used in the proof. The {\em size} of a polynomial $p$ is the number of  terms in $p$.  The {\em size} a PC derivation $p_1,\ldots, p_m$  is the sum of the sizes of the polynomials $p_1,\ldots,  p_m$.

\medskip

For our upper bound in Theorem~\ref{thm:upper}, we also consider another algebraic proof system, known as {\em Nullstellensatz} (NS), to certify unsolvability of sets of polynomial equations. Nullstellensatz  is defined in a static form as follows:  a  refutation of a list  ${\cal P}=(p_1,\ldots,p_m)$ of polynomial equations over variables $x_1\ldots,x_n$ 
is given by the list of polynomials  ${\cal Q}=(q_1,\ldots q_m)$  such that
$$
\sum_{i \in [m]} p_iq_i =1
$$

The {\em degree} of a NS refutation is the maximal degree of a polynomial in ${\cal P} \cup {\cal Q}$. The size of $NS$ proof is the sum of the number of monomials appearing in the polynomials
$q_1,\ldots, q_m$.

\medskip

{\em Sum-of-Squares} (SOS) is a static proof system for certifying the unsolvability of systems of polynomial equations and polynomial inequalities, where polynomials are usually over the ring $\mathbb R[x_1\ldots, x_n]$.    

A polynomial $p$ is a {\em sum-of-squares} polynomial if it is in the form $p=\sum_{i} r_i^2$ and the $r_i$'s  are polynomials as well.  Given a system made by a set of polynomial equations  
 ${\cal P} =\{p_1=0,\ldots p_m=0\}$ and a set ${\cal Q}=\{q_1\geq 0,\ldots q_k\geq 0\}$ of polynomial inequalities, a {\em sum-of-squares} proof of the polynomial inequality $p\geq 0 $ from ${\cal P} \cup {\cal Q}$  is  given by the  formal identity
 $$
 p= s_0+\sum_{i\in [k]}s_iq_i + \sum_{j \in [m]}t_jp_j
 $$
where $s_0,s_1,\ldots, s_k$ are sum-of-squares polynomials, while $t_1,\ldots,t_m$ are arbitrary polynomials.  When the system  ${\cal P} \cup {\cal Q}$  is unsatisfiable, a {\em refutation} of ${\cal P} \cup {\cal Q}$ is a proof of the inequality $-1\geq 0$, that is for  $p$ the constant polynomial $-1$. The {\em degree} of the proof is 
the $\max \{\deg(p),\deg(s_0), \deg(s_i)+\deg(q_i), \deg(t_j)+\deg(p_j) | i \in [k], j \in[m]\}$.

\medskip 

\begin{definition}[PC reduction between systems of polynomials, cf. {\cite[Sec.~3]{DBLP:journals/jcss/BussGIP01}}] \label{def:reduction}
Let $P(x_1,\ldots,x_n)$ and $Q(y_1,\ldots, y_m)$ be two sets of polynomials over a field $\mathbb F$.  
$P$ is \emph{$(d_1,d_2)$-reducible} to $Q$ if:
\begin{enumerate}
\item For each $i \in [m]$ there is a polynomial $r_i(\mathbf{x})$ of degree at most $d_1$  (which we think of as defining $y_i$ in terms of the $\mathbf x$ variables);
\item There exists a degree $d_2$ PC derivation of $Q(r_1(\mathbf x),\ldots, r_m(\mathbf x) )$ from polynomials $P(\mathbf x)$.
\end{enumerate}
\end{definition}

\begin{lemma}[{\cite[Lem.~1]{DBLP:journals/jcss/BussGIP01}}]
\label{lem:PCred}
If $P(\mathbf{x})$ is $(d_1, d_2)$-reducible to $Q(\mathbf{y})$ and there is
a degree $d$ PC refutation of $Q(\mathbf{y})$, then there is a degree $\max(d_2, d_1 d)$ 
refutation of $P(\mathbf x)$.
\end{lemma}

In their paper, they typically only applied this to systems of equations which were known to be unsatisfiable (such as PHP and Tseitin tautologies), whereas in our paper we have several situations we want to combine the above notion together with the usual notion of many-one reduction. We encapsulate this in the following definition. We say a decision problem $\Pi$ is a \emph{polynomial solvability problem} over a field $\F$ if all valid instances of the problem are systems of polynomial equations over $\F$, and the problem is to decide whether such a system of equations has solutions over the algebraic closure $\overline{\F}$. Thus, the difference between multiple polynomial solvability problems is just \emph{which} systems of equations are valid inputs.

\begin{definition}[PC many-one reduction] \label{def:reduction2}
Let $\Pi_1, \Pi_2$ be two polynomial solvability problems over a field $\F$. We say that $\Pi_1$ \emph{$(d_1,d_2)$-many-one reduces} to $\Pi_2$ if there is a polynomial-time many-one reduction $\rho$ from $\Pi_1$ to $\Pi_2$, such that for all unsatisfiable instances $\mathcal{F}$ of $\Pi_1$, $\mathcal{F}$ $(d_1, d_2)$-reduces to $\rho(\mathcal{F})$. When this occurs with $d_1,d_2 = O(1)$, we write \[\Pi_1 \leq_m^{PC} \Pi_2.\]
\end{definition}

\subsection{Linear algebra and tensors}
Given three vector spaces $U,V,W$ over a field $\F$, a 3-tensor is an element of the vector space $U \otimes V \otimes W$, whose dimension is $(\dim U)(\dim V)(\dim W)$. If $e_i$ is the $i$-th standard basis vector, then a basis for $U \otimes V \otimes W$ is given by the vectors $\{e_i \otimes e_j \otimes e_k\}$. One may also interpret the symbol $\otimes$ more concretely as the Kronecker product, in which $e_i \otimes e_j \otimes e_k$ represents a 3-way array whose only nonzero entry is in the $(i,j,k)$ position. The vector space of such 3-way arrays (with coordinate-wise addition) is isomorphic to $U \otimes V \otimes W$.

The rank of a tensor $T \in U \otimes V \otimes W$ is the minimum $r$ such that $T = \sum_{i=1}^r u_i \otimes v_i \otimes w_i$ for some vectors $u_i, v_i, w_i$. 

Two $n \times m \times p$ 3-tensors $T,T' \in U \otimes V \otimes W$ are isomorphic if there exist matrices $X \in \GL(U), Y \in \GL(V), Z \in \GL(W)$ such that $(X,Y,Z) \cdot T = T'$, where the latter is shorthand for (\ref{eq:TI}).
If we treat $T,T'$ as given non-isomorphic tensors, then we may treat (\ref{eq:TI}) as a system of equations in the $n^2 + m^2 + p^2$ variables $X_{ii'}, Y_{jj'}, Z_{kk'}$. To enforce that these variable matrices are invertible, we furthermore introduce three additional sets of variables $X', Y', Z'$ meant to be the inverse matrices, and include also the equations 
\[
XX' = X'X = I_n \qquad YY' = Y'Y = I_m \qquad ZZ' = Z' Z = I_p,
\]
where $I_n$ denotes the $n \times n$ identity matrix, which is $\Id_U$ in any basis. (We could have instead introduced new variables such as $\delta$ and the equation $\det(X) \delta = 1$, however, the latter equation is degree $n$, whereas the above equations all have degree $O(1)$, which is more desirable from the point of view of algebraic proof complexity.)

\subsection{Polynomial encodings and the inversion principle} \label{sec:prelim:inv}
Some  principles of linear algebra can be formulated as tautologies in propositional logic and therefore also as a set of polynomial equations.
In this paper we preliminarily consider two such principles. 

\medskip

\noindent {\bf Rank Principle.} As a first example we consider a set of unsatisfiable polynomials encoding the principle that  the product of a $n\times r$ matrix $X$ by a $r\times n$ matrix $Y$ cannot be the identity matrix whenever $r<n$.  We consider variables $x_{i,k},y_{j,k}$ for $i,j\in [n]$ and $k\in [r]$, where $r< n$ to encode $X$ and $Y$. Then the polynomial encoding is:
$$\mathbb I(r,n) := \sum_{k\in[r]} x_{i,k}y_{j,k}-\delta_{i,j} \quad \quad i,j\in [n]$$
where $\delta_{i,j}=1$ if $i=j$ and $0$ otherwise. This set of polynomials is clearly unsatisfiable as long as $r<n$. 

\medskip

\noindent {\bf Inversion Principle.} The second principle encodes the invertibility of a square $n\times n$ matrix $A$,  expressing the tautology that
 $AB=I \rightarrow BA=I$ where $A,B$ are $n\times n$ matrices and $I$ is the identity matrix. Stephen A. Cook suggested this principle as a tautology that may be hard to prove in several proof systems. 
 
Let $a_{i,j},b_{i,j}$ be formal variables 
 encoding respectively  the $(i,j)$-th entries of $A$ and $B$.  We represent the fact that $AB=I$ as the set of degree $2$ polynomials

$$
\sum_{k \in [n]} a_{i,k}b_{k,j} - \delta_{i,j}\qquad i,j \in [n],
$$
where $\delta_{i,j}=1$ if $i=j$ and $0$ otherwise. We denote this set of polynomials by $AB=I$. In Section \ref{sec:matrix}, we study the degree complexity of  $AB=I \vdash BA=I$, that is of PC derivations of the polynomials $BA=I$ from the polynomials $AB=I$. 

\medskip

In view of the results we obtain in Section \ref{sec:matrix}, in Section~\ref{sec:conj} we considered a  {\em polynomial rule schema} of the form 
$$
\frac{AB=I}{BA=I}
$$
which we call the {\em Inversion Rule} (INV) meant to be added to $PC$ as an extra rule. We make this slightly more precise here.

A polynomial instantiation $\tau$ of the polynomials $AB=I$ is a substitution of polynomials $p_{i,j},q_{i,j}$ to variables $a_{i,j}$ and $b_{i,j}$.   
In PC+INV a polynomial $p$ is derivable from a set of polynomials
${\cal P}$ if 
\begin{enumerate}
\item $p$ is an axiom, or $p \in {\cal P}$;
\item $p$ is obtained by multiplication or linear combination from previous polynomials in the proof;
\item $p$ is a polynomial among a polynomial instantiation $\tau$ of $BA=I$, given that among the polynomials previously derived  in the proof there are all the polynomials 
forming the instantiation $\tau$ of $AB=I$.
\end{enumerate}

\medskip

\noindent {\bf Pigeonhole Principle.} An important role in proving the results in Section \ref{sec:matrix} is played by the well-known {\em Pigeonhole principle} stating that  any function $f$ from $[n]$ to $[r]$ with $r<n$ has a collision, that is there are $i \not= i' \in [n]$ and a $j\in[r]$ such that $f(i)=f(i')=j$. $PHP^n_r$ is the set of polynomials:

$$
\sum_{k\in [r]}p_{i,k} -1, \mbox{ for } i \in [n], \qquad \qquad p_{i,k}p_{j,k}, \mbox{ for } i \not= j \in [n], k \in [r] \qquad p_{ij}^2 - p_{ij} \text{, for $i \in [n], j \in [r]$} 
$$

Razborov \cite{DBLP:journals/cc/Razborov98} additionally included the ``functional equations'' (encoding that each pigeon cannot be matched to more than one hole):
\[
p_{i,k} p_{i,k'}, \mbox{ for } i \in [n], k \neq k' \in [r].
\]

\medskip

%% file: matrices.tex

\section{Linear algebra warm-up: PC for matrices}
\label{sec:matrix}
Two matrices $M,M' \in U \otimes V$ are isomorphic as tensors if they are equivalent as matrices, meaning under left- and right-multiplication by invertible matrices $X \in \GL(U), Y \in GL(V)$, that is, 
\[
XMY = M'. 
\]

Since we want $X,Y$ to be invertible, we also introduce variable matrices $X', Y'$ as before, together with the equations
\[
XX' = X'X = \Id_U \qquad YY' = Y'Y = \Id_V.
\]
Then by left multiplying our initial matrix equation by $Y'$, we may replace it with the new matrix equation
\[
XM = M'Y'.
\]
The latter has the advantage of being linear in $X$ and $Y'$, but the quadratic equations $XX' =  \Id_U, YY' =  \Id_V$ still make even this case not totally obvious.

\subsection{A trick for PC degree} \label{sec:trick}
If our focus is on PC \emph{degree}, we note that the degree of the equations is unchanged if we first left- or right-multiply $M,M'$ by invertible scalar matrices. For example, if we replace $M$ by $\overline{M} = AMB$ with $A,B \in \GL(U)$, then we may replace $X$ by $\overline{X} := XA^{-1}$, $Y$ by $\overline{Y} := B^{-1} Y$. Then we have $\overline{M} \cong M$, so $\overline{M} \cong M'$ iff $M \cong M'$. Furthermore, since the transformation $X \mapsto XA^{-1}$, $Y \mapsto B^{-1} Y$ is linear and invertible, any PC proof that $\overline{M} \not\cong M'$ can be transformed by the inverse linear transformation into a PC proof that $M \not\cong M'$ of the same degree.

Now, for matrices under this equivalence relation, we have a normal form, namely every matrix $M$ is equivalent to a diagonal matrix with $\rk(M)$ 1s on the diagonal and all the remaining entries $0$, that is, $\sum_{i=1}^{\rk(M)} e_i \otimes e_i = I_r \oplus 0$, where the latter $0$ denotes a $0$ matrix of appropriate size $(n-r) \times (m-r)$. So by using the preceding trick, we may put both $M$ and $M'$ in this form. The two are isomorphic iff $\rk(M) = \rk(M')$, so for PC degree we have now reduced to the case of showing that $I_r \oplus 0$ and $I_{r'} \oplus 0$ are not isomorphic when $r \neq r'$. 

Note that, aside from the equations saying $X$ and $Y$ are invertible, this is almost identical to the Rank Principle (see Section~\ref{sec:prelim:inv}). In the rest of this section we will prove PC lower bounds on both the Rank Principle and the Inversion Principle. Here, we show that the addition of these invertibility axioms in fact makes 2TI much easier in PC than the Rank Principle or 3TI.

\begin{proposition}
Let $M,M'$ be two $n \times m$ matrices of  ranks $r,r'$ respectively, with $r' > r$. Then, over any field whose characteristic does not divide $r'-r$, the following equations have a degree 3 PC refutation and a degree 4 NS refutation:
\[
XMY^T = M' \qquad XX' = X' X = \Id_n \qquad YY' = Y' Y = \Id_m.
\]
\end{proposition}

For those familiar with the low-degree PC proof of the functional onto-PHP, the following proof is similar. 
\newcommand{\Tr}{\text{Tr}}

\begin{proof}[Proof idea]
By the observations in Section~\ref{sec:trick}, we may assume without loss of generality (from the point of view of PC degree) that $M = \Id_r \oplus 0_{n - r \times m-r}$ and $M' = \Id_{r'} \oplus 0_{n-r' \times m-r'}$. 

Write $X = \begin{bmatrix} X_{11} & X_{12} \\ X_{21} & X_{22} \end{bmatrix}$ where the top-left block $X_{11}$ has size $r' \times r$, and similarly write $Y = \begin{bmatrix} Y_{11} & Y_{12} \\ Y_{21} & Y_{22} \end{bmatrix}$ where $Y_{11}$ has size $r' \times r$. In this notation, the matrix equation $XMY^T = M'$ becomes the equations
\begin{align*}
XMY^T & = \begin{bmatrix} X_{11} & X_{12} \\ X_{21} & X_{22} \end{bmatrix} \begin{bmatrix} \Id_r \\ & 0_{(n-r) \times (m-r)} \end{bmatrix} \begin{bmatrix} Y_{11}^T & Y_{21}^T \\ Y_{12}^T & Y_{22}^T  \end{bmatrix} \\
 & = \begin{bmatrix} X_{11} & 0 \\ X_{21} & 0 \end{bmatrix} \begin{bmatrix} Y_{11}^T & Y_{21}^T \\ 0 & 0 \end{bmatrix} \\
 & = M' = \begin{bmatrix} \Id_{r'} \\ & 0_{(n-r') \times (m-r')} \end{bmatrix}
\end{align*}
which becomes the four matrix equations
\begin{equation} \label{eq:2}
X_{11} Y_{11}^T = \Id_{r'} \qquad X_{11}Y_{21}^T = 0 \qquad X_{21}Y_{11}^T = 0 \qquad X_{21}Y_{21}^T = 0.
\end{equation}
Note that so far our PC proof hasn't actually done anything---it is all just notation, and all in the same degree we started with (degree 2).

Then, using the equations $XX' = \Id_n$ and $YY'=\Id_m$, we will derive that $Y_{11}^T X_{11} = \Id_r$. Then we derive 1 as 
\[
\frac{1}{r-r'}\left(\Tr(X_{11} Y_{11}^T - \Id_{r'}) - \Tr(Y_{11}^T X_{11} - \Id_r)\right).
\]
The point here is that trace is additive and cyclically invariant, so $\Tr(X_{11} Y_{11}^T) \equiv \Tr(Y_{11}^T X_{11})$, identically as polynomials, so there is no further derivation needed.
\end{proof}

\begin{proof}
The proof starts using the first part of the proof idea above, so we continue from Equation~(\ref{eq:2}) with the notation introduced above.
In the remainder of the proof, we will derive $Y_{11}^T X_{11} = \Id_r$. Then the last paragraph of the proof idea will complete the proof.

To derive $Y_{11}^T X_{11} = \Id_r$, we will use the invertibility equations (those involving $X'$ and $Y'$). Write $X' = \begin{bmatrix} X'_{11} & X'_{12} \\ X'_{21} & X'_{22} \end{bmatrix}$, where $X'_{11}$ has size $r \times r'$ (NB: the size is the ``transpose'' of the size of $X_{11}$) and similarly for $Y'$. 

From considering the upper-left $r \times r$ block of the matrix equation $X' X = \Id_n$, we get
\[
X'_{11} X_{11} + X'_{12} X_{21} = \Id_r.
\]
Right multiplying by $Y_{11}^T$, we get
\[
X'_{11} X_{11} Y_{11}^T + X'_{12} X_{21} Y_{11}^T = Y_{11}^T. 
\]
But now we can subtract from this $X'_{11}$ times the equation $X_{11} Y_{11}^T = \Id_{r'}$, and also $X'_{12}$ times the equation $X_{21} Y_{11}^T = 0$ to get
\begin{equation}
X'_{11} = Y_{11}^T.
\end{equation}

Similarly, considering the upper-left $r \times r$ block of the matrix equation $Y' Y = \Id_m$, we get $Y'_{11} Y_{11} + Y'_{12} Y_{21} = \Id_{r}$. For consistency with the notation above, we take the transpose of this entire equation (in PC, this is essentially a null-op---we are just re-arranging how we are viewing a set of $(r')^2$ equations on the page), to get:
\[
Y_{11}^T (Y'_{11})^T + Y_{21}^T (Y'_{12})^T = \Id_r.
\]
Left multiplying by $X_{11}$, we get
\[
X_{11} Y_{11}^T (Y'_{11})^T + X_{11} Y_{21}^T (Y'_{12})^T = X_{11}.
\]
Now, right-multiplying the equation $X_{11} Y_{11}^T = \Id$ by $(Y'_{11})^T$, and right-multiplying the equation $X_{11} Y_{21}^T = 0$ by $(Y'_{12})^T$ and subtracting both of these from the above, we get
\begin{equation}
(Y'_{11})^T = X_{11}.
\end{equation}

Next, we derive $M - X' M' (Y')^T=0$, as follows: left-multiplying $XMY^T - M'$ by $X'$, and subtract from it $(X'X-I)$ times $MY^T$, to get $-X'M' + MY^T$. Now right-multiply the latter by $(Y')^T$ and subtract from it $M$ times $(Y^T (Y')^T - I)$, yielding $-X'M'(Y')^T + M$. Now multiply by $-1$.

Now, from $X' M' (Y')^T = M$, as at the beginning of the proof, we derive that $X'_{11} (Y'_{11})^T = \Id_r$. But above we have derived that $X'_{11} = Y_{11}^T$ and $(Y'_{11})^T = X_{11}$, so from the preceding three equations we get $Y_{11}^T X_{11} = \Id_r$, as claimed. This completes the PC proof.

Let us unroll the PC proof to derive a Nullstellensatz proof (here we underline the use of original equations):
\begin{align*}
r - r' & = \Tr(\underline{X_{11} Y_{11}^T - \Id_{r'}}) - \Tr(Y_{11}^T X_{11} - \Id_r)  \\
\end{align*}
Now we focus on the NS derivation of $Y_{11}^T X_{11} - \Id_r$. Since the trace is linear, and we are focusing on degree, this is without loss of generality. We have:
\begin{align*}
Y_{11}^T X_{11} - \Id_r = & \left(X'_{11} (Y'_{11})^T - \Id_r\right) - (X'_{11} - Y_{11}^T) (Y'_{11})^T -  Y_{11}^T \left((Y'_{11})^T - X_{11}\right) \\
 = & \left(-X' \underline{(XMY^T - M')} Y'^T + \underline{(X'X - I)}MY^T (Y')^T + M \underline{(Y^T (Y')^T - I)}\right)_{11} \\
  & + \left(\underline{(X'_{11} X_{11} + X'_{12} X_{21} - \Id_r)} Y_{11}^T - X'_{11} \underline{(X_{11} Y_{11}^T - \Id_{r'})} - X'_{12}\underline{(X_{21} Y_{11}^T)} \right) (Y'_{11})^T \\ 
  & -  Y_{11}^T \left( X_{11} \underline{(Y_{11}^T (Y'_{11})^T + Y_{21}^T (Y'_{12})^T - \Id_r)} - \underline{(X_{11} Y_{11}^T - \Id)} (Y'_{11})^T - \underline{(X_{11} Y_{21}^T)} (Y'_{12})^T   \right). 
\end{align*}
This is visibly degree 4.
\end{proof}

\subsection{Inversion Principle implies the Rank Principle}

\begin{lemma} \label{lem:inv_to_rank}
If the $r \times r$ Inversion Principle has a degree $d$ PC derivation, 
then there is a degree $\max\{d,3\}$ PC refutation of the Rank Principle stating that a rank $r$ matrix is not equivalent (isomorphic) to a rank $n$ matrix, for any $n > r$. 

If the Inversion Principle has a degree $d$ NS derivation, then the Rank Principle has a degree $d+2$ NS refutation.
\end{lemma}

\begin{proof}
Suppose the $r \times r$ Inversion Principle has a degree-$d$ derivation. Consider the Rank Principle $XY=I_n$ where $X$ is $n \times r$ and $Y$ is $r \times n$, with $n > r$. Write
\[
X = \begin{bmatrix} X_0 \\ X_1 \end{bmatrix} \text{ and } Y = \begin{bmatrix} Y_0 & Y_1 \end{bmatrix},
\]
where $X_0,Y_0$ are $r \times r$. Then, examining the upper-left $r \times r$ corner of the original equations, we find $X_0 Y_0 = I_r$. As these are square matrices, by assumption in degree $d$ we may then derive that $Y_0 X_0 = I_r$ as well.

Now, multiply both sides of $XY=I_n$ on the left by the matrix $\begin{bmatrix} Y_0 & 0 \\ 0 & I_{n-r} \end{bmatrix}$. The result is then the set of degree-3 equations
\[
\begin{bmatrix} Y_0 X_0 \\ X_1 \end{bmatrix} \begin{bmatrix} Y_0 & Y_1 \end{bmatrix} = \begin{bmatrix} Y_0 & 0 \\ 0 & I_{n-r} \end{bmatrix}.
\]
Considering the upper-right $r \times (n-r)$ block of these equations, we find the equations $Y_0 X_0 Y_1 = 0$. 

But now, from the equation $Y_0 X_0 = I_r$, we may right-multiply by $Y_1$ to get $Y_0 X_0 Y_1 = Y_1$. Combining with the equation at the end of the last paragraph, we then conclude $Y_1 = 0$.

Finally, consider the lower-right $(n-r) \times (n-r)$ part of the original equation $XY = I_n$, namely, $X_1 Y_1 = I_{n-r}$. We had already derived $Y_1 = 0$, which we can then left-multiply by $X_1$ to get $X_1 Y_1 = 0$. Considering any diagonal entry of these two equations, we then derive the contradiction $1=0$.

To see the NS certificate, we unwrap the above proof. First write $Y_0 X_0 - I_r$ as a linear combination of the equations $X_0 Y_0 - I_r$ with polynomial coefficients, in total degree $d$. Among our starting equations in the Rank Principle, we have $X_0 Y_1$ and $X_1 Y_1 - I_{n-r}$. Then the following linear combination has degree 2 more than $Y_0 X_0 - I_r$, and derives $1$ in any of its diagonal entries: 
\[
-X_1 Y_0 \underline{X_0 Y_1} + X_1 (Y_0 X_0 - I_r) Y_1 + \underline{(X_1 Y_1 - I_{n-r})}.
\]
\end{proof}

\begin{observation}
The $n \times n$ Inversion Principle has a proof of degree $2n+2$.
\end{observation}

\begin{proof}
The idea is to use Laplace expansion. We spell out the details.

We start with $XY = I_n$, where $X$ and $Y$ are $n \times n$ matrices of variables. Left-multiply by $Y$ to get $YXY = Y$, and then right multiply by $Adj(Y)$ (whose entries are the $(n-1) \times (n-1)$ cofactors of $Y$, hence have degree $n-1$) to get $YXYAdj(Y) = YAdj(Y)$. Now, by Laplace expansion, we have $Y Adj(Y) \equiv \det(Y) I_n$, so we get $YX \det(Y) = \det(Y) I_n$.

Next, starting from $XY=I_n$ and expanding out the determinant term-by-term, we derive $\det(XY) = 1$. (Note that here, we are not simply applying the determinant to the matrix $XY-I$, as that would give us the value of the characteristic polynomial evaluated at 1. Instead, we repeatedly use that from $a-b=0$ and $c-d=0$ we can derive $ac-bd=0$ as $\underline{(a-b)}c + b\underline{(c-d)}$. Similarly, we can derive $(a+c)-(b+d)=0$ as $(a-b) + (c-d)$.) Now, since $\det(XY) \equiv \det(X)\det(Y)$ identically as polynomials, we have derived $\det(X)\det(Y)=1$ in degree $n$.

Now, from $YX \det(Y) - \det(Y) I_n$ in the first paragraph, we multiply by $\det(X)$ to get $(YX - I_n)(\det(X)\det(Y))$. From $\det(X)\det(Y)-1$ in the second paragraph, we multiply by $-(YX-I_n)$ and add to the preceding to get $YX - I_n$, all in degree at most $2n+2$.
\end{proof}

\subsection{Lower bound on the Rank Principle (and Inversion Principle) via reduction from PHP}
\newcommand{\bbI}{\mathbb{I}}

Here we show that the Rank Principle (see Section~\ref{sec:prelim:inv}) requires large PC degree, via a reduction to the Pigeonhole Principle. 
For the Pigeonhole principle, a tight PC degree lower bound is known:

\begin{theorem}[{Razborov \cite{DBLP:journals/cc/Razborov98}}]
\label{thm:PCdeglb}
Any $PC$ refutation of the Functional $PHP^n_r$ requires degree $r/2$+1 over any field.
\end{theorem}

We use this to show:

\begin{theorem} \label{thm:rank}
Let $n\in \mathbb N$,  $n\geq 2$ and $1\leq r<n$. $\bbI(r,n)$ (with or without the Boolean axioms) requires degree $r/2+1$ in $PC$ over any field.
\end{theorem}
\begin{proof}
We prove that $PHP^n_{r}$ is $(1,2)$-reducible to $\bbI(r,n)$. 
First we consider the following degree $1$  polynomials  defining $x$ and $y$ variables of $\bbI(r,n)$ 
in terms of the $p$ variables of $PHP^n_{r}$.
variables 
$$
x_{i,k} = y_{i,k} = p_{i,k} \qquad \mbox{ for } i \in [n], k \in [n-1].
$$ 
Second we show a degree $2$ $PC$ proof   of $\bbI(r,n)$ from the  polynomials defining the  $PHP^n_r$.
From $PHP$ axioms $p_{i,k}p_{k,j}$ for $i,j \in [n], i\not =j$,  and summing over all $k \in [r]$, we get 
$$
\sum_{k \in [r]}p_{i,k}p_{k,j},
$$ 
which are exactly the axioms of $\bbI(r,n)$ for $i \not =j, i,j \in [n]$, after the substitution of variables.

For a $i\in [n]$, take  the boolean axioms written in the form $p_{i,k}p_{i,k}-p_{i,k}$ and  sum them over $k \in [r]$:
$$\sum_{k \in [r]} p_{i,k}p_{i,k} - \sum_{k \in [r]}p_{i,k}$$
Summing this last polynomial with the $PHP$ axiom $\sum_{k \in [r]} p_{i,k} -1$
we get the polynomial
$$\sum_{k\in [r]}p_{i,k}p_{i,k} -1,$$
which is the axiom of $\bbI(r,n)$ for $i=j$ after the substitution of the  variables.  The proof has degree $2$. 
The result follows immediately  from Lemma \ref{lem:PCred} and Theorem \ref{thm:PCdeglb}. 

\end{proof}

\begin{corollary}
Any $PC$ proof of $AB=I \vdash BA=I$, where $A,B$ are square $n\times n$ $\{0,1\}$ matrices requires degree $n/2+1$.
\end{corollary}

\begin{proof}
Follows immediately from Theorem~\ref{thm:rank} and Lemma~\ref{lem:inv_to_rank}.
\end{proof}

%% file: tensors.tex

\section{Upper bound for non-isomorphism of bounded-rank tensors} \label{sec:upper}
\begin{theorem} \label{thm:lowrank}
Over any algebraically closed field, there is a function $f(r) \leq 2^{O(r^2)}$, depending only on $r$, such that, given two non-isomorphic tensors $M,M'$ of tensor rank $\leq r$, the Nullstellensatz degree of refuting isomorphism is at most $f(r)$.

If working over a finite field $GF(q)$ and including the equations $x^q - x = 0$ for all variables $x$, then the PC degree is at most $12qr^2$.
\end{theorem} 

\begin{proof}
The proof is based mainly on the so-called inheritance property of tensor rank.

Let $M = \sum_{i=1}^r u_i \otimes v_i \otimes w_i$ and let $M' = \sum_{i=1}^r u_i' \otimes v_i' \otimes w_i'$ be our two tensors of format $n_1 \times n_2 \times n_3$. Let $d_1 = \dim \Span\{u_1, u_2, \dotsc, u_r, u_1', u_2', \dotsc, u_r'\}$, $d_2$ similarly for the $v$'s and $d_3$ for the $w$'s. Choose a basis $e_1, e_2, \dotsc, e_{n_1}$ for $\F^{n_1}$ such that $\Span\{e_1, \dotsc, e_{d_1}\} = \Span\{u_1, \dotsc, u_r, u_1', \dotsc, u_r'\}$. Let $f_1, \dotsc, f_{n_2}$ be a similar basis for $\F^{n_2}$ (with the first $d_2$ vectors a basis for $\Span\{v_1, \dotsc, v_r, v_1', \dotsc, v_r'\}$), and similarly $g_1, \dotsc, g_{n_3}$. Changing everything in sight into the $e_{\bullet} \otimes f_{\bullet} \otimes g_{\bullet}$ basis, we find that $M,M'$ are both supported in the upper-left $d_1 \times d_2 \times d_3$ sub-tensors, with all zeros outside of this. Call the corresponding $d_1 \times d_2 \times d_3$ tensors $\overline{M}, \overline{M}'$. Because all the entries outside this box are zero, it is not difficult to show that $M \cong M'$ iff $\overline{M} \cong \overline{M}'$ (the so-called ``Inheritance Theorem,'', see, e.\,g., \cite[\S 3.7.1]{landsbergBook}); note that isomorphism of $\overline{M}$ with $\overline{M}'$ is via the much smaller group $\GL_{d_1} \times \GL_{d_2} \times \GL_{d_3}$, rather than $\GL_n \times \GL_n \times \GL_n$ (the latter of which is used to determine isomorphism of $M$ with $M'$).

In this basis, isomorphism of $\overline{M}, \overline{M}'$ is solely determined by the upper-left $d_1 \times d_1$ sub-matrix of $X, X'$, the upper-left $d_2 \times d_2$ submatrix of $Y, Y'$, and the upper-left $d_3 \times d_3$ sub-matrix of $Z, Z'$. So we now only need to deal with equations in $d_1^2 + d_2^2 + d_3^2$ variables. Since each $d_i \leq 2r$, this is at most $12r^2$ variables.

Since we have $\leq 12r^2$ variables, $d_1 d_2 d_3$ cubic equations, and $6n^2$ quadratic equations ($XX' = I = X'X = YY' = \dotsb$), over an algebraically closed field Sombra's Effective Nullstellensatz \cite{sombra} implies that the Nullstellensatz degree of refuting our equations is then at most $4 \cdot 3^{\Theta(r^2)}$. 

Over a finite field with the extra equations $x^q = x$, we may reduce degrees so that the degree of each variable is never more than $q$, the size of the field. In this case, the PC degree is at most $q$ times the number of variables, i.\,e., at most $12 q r^2$.
\end{proof}

\input{GI.tex}

\section{Lower bound on PC degree for Tensor Isomorphism from Random 3XOR} \label{sec:lowerbound}
We get a lower bound on PC refutations for \textsc{Tensor Isomorphism} through the following series of low-degree PC many-one reductions (Definition~\ref{def:reduction2}):
\begin{align}
\textsc{Random 3-XOR} 
 \leq_m^{PC} & \{\pm 1\}\textsc{-Monomial Equivalence of} \\
 & \{\pm 1\}\textsc{-Multilinear Noncommutative Cubic Forms} \\ 
 \leq_m^{PC} & \textsc{Monomial Equivalence of } \{\pm 1\}\textsc{ Noncommutative Cubic Forms} \\
 \leq_m^{PC} & \textsc{Equivalence of } \{\pm 1\}\textsc{ Noncommutative Cubic Forms} \label{reduction:NC} \\
 \leq_m^{PC} & \textsc{Tensor Isomorphism} \label{reduction:TI}
\end{align}
We then appeal to the following  PC lower bound on \textsc{Random 3-XOR}:

\begin{theorem}[{Ben-Sasson \& Impagliazzo \cite[Thm.~3.3 \& Lem.~4.7]{BI}}] \label{thm:random}
Let $\F$ be any field of characteristic $\neq 2$. A random 3-XOR instance with clause density $\Delta = m/n$ requires PC degree $\Omega(n/\Delta^2)$ to refute, with probability $1-o(1)$.
\end{theorem}

This allows us to prove:

\begin{theorem} \label{thm:lowerbound}
Over any field of characteristic $\neq 2$, there is a random distribution of instances of $n \times n \times n$ \textsc{Tensor Isomorphism}---which assigns nonzero probability to at least $2^{\Omega(\sqrt[4]{n}) \log n}$ different instances---whose associated equations require PC degree $\Omega(\sqrt[4]{n})$ to refute, with probability $1-o(1)$.
\end{theorem}

Note that such instances have $N=6n^2$ variables, so this is really only an $\Omega(\sqrt[8]{N})$ lower bound relative to the number of variables.

In the following subsections we recall the definitions of the above problems and their associated systems of polynomial equations, and we give the reductions in the order listed above.

The first two reductions are gadget constructions of linear size; the proof of correctness for the first uses the fact that random hypergraphs have no automorphisms, while the second is fairly straightforward. Reduction (\ref{reduction:NC}) uses a gadget from Grochow \& Qiao \cite{GQ}, albeit for a new application, and shows that the reduction using this gadget also yields a low-degree PC reduction. Reduction (\ref{reduction:TI}) is based on two lemmas, which show that the many-one reduction for this problem in fact also gives a low-degree PC reduction. 

\begin{remark}
Both of the latter two reductions have a quadratic size increase, so while we get a nearly-linear lower bound on PC degree for refutations of \textsc{Monomial Equivalence of Noncommutative Cubic Forms}, we only get a $\Omega(\sqrt{n})$ degree lower bound \textsc{Equivalence of Noncommutative Cubic Forms} and a $\Omega(\sqrt[4]{n})$ degree lower bound on \textsc{Tensor Isomorphism}. If the gadget sizes of these latter two reductions could be improved to linear, we would get a similarly near-linear lower bound (linear in the side length, still $\sqrt{N}$ relative to the number of variables) on PC refutations for \textsc{Tensor Isomorphism} as well. As many of the reductions in \cite{FGS, GQ} are of a similar flavor to the ones we consider here, we believe that they can all be proven in low-degree PC, so we expect the main obstacle to such an improvement is the size of the constructions themselves.
\end{remark}

\subsection{From \textsc{Random 3-XOR} to $\{\pm 1\}$-multilinear noncommutative cubic forms}
\begin{definition}
A random 3-XOR instance with $n$ variables and $m$ clauses is obtained by sampling $m$ clauses independently and uniformly from the set of all $2 \binom{n}{3}$ parity constraints on 3 variables. Each parity constraint is encoded by an equation of the form $x_i x_j x_k = \pm 1$, and the Boolean constraints are encoded by $x_i^2 = 1$.  
\end{definition} 

By a $\{\pm 1\}$-monomial matrix, we mean a monomial matrix in which all nonzero entries are one of $\pm 1$. $\{\pm 1\}$\textsc{-Monomial Equivalence of Noncommutative Cubic Forms} is the problem of deciding, given two noncommutative cubic forms $f,f'$ in $n$ variables $x_1, \dotsc, x_n$ with all nonzero coefficients $\pm 1$, whether there is a permutation $\pi \in S_n$ and signs $e_i \in \{\pm 1\}$ such that $f(e_1 x_{\pi(1)}, \dotsc, e_2 x_{\pi(2)}, \dotsc, e_n x_{\pi(n)}) = f'(\vec{x})$. Equivalently, if we represent a noncommutative cubic form $f$ by the 3-way array $T_{ijk}$ such that $f(\vec{y}) = \sum_{i,j,k \in [n]} T_{ijk} y_i y_j y_k$, the problem here asks whether there is a $\{\pm 1\}$-monomial matrix $A$ such that $(A,A,A) \cdot T = T'$, that is, whether $T'_{i'j'k'} = \sum_{ijk} a_{ii'} a_{jj'} a_{kk'} T_{ijk}$ for all $i',j',k' \in [n]$.

\begin{definition}\label{def:equations for equivalence}
We define the systems of equations associated to several variations of \textsc{Equivalence of Noncommutative Cubic Forms}.
\begin{enumerate}
\item Given two $n \times n \times n$ 3-way arrays $T,T'$, the system of equations for \textsc{Equivalence of Noncommutative Cubic Forms}  is the following system of equations in $2n^2$ variables. Let $A,A'$ be $n \times n$ matrices of independent variables $a_{ij}, a'_{ij}$, respectively.
\[
\begin{array}{lll}
(A,A,A) \cdot T = T'  & & \text{($A$ is an equivalence)} \\
AA' = A'A =\Id & & \text{($A$ is invertible with $A^{-1} = A'$)} \\
\end{array}
\]

\item The system of equations for \textsc{Monomial Equivalence of Noncommutative Cubic Forms} includes the preceding equations, as well as:
\[
\begin{array}{lll}
a_{ij} a_{ij'} = 0 & \forall i \forall j \neq j' & \text{(at most one nonzero per row)} \\
a_{ij} a_{i'j} = 0 & \forall j \forall i \neq i' & \text{(at most one nonzero per column)}
\end{array}
\]

\item The system of equations for  $\{\pm 1\}$\textsc{-Monomial Equivalence of Noncommutative Cubic Forms} includes all the preceding equations, as well as
\[
\begin{array}{lll}
a_{ij} (a_{ij} +1)(a_{ij}-1)=0 & \forall i,j \in [n] & \text{(all entries in $\{0,\pm 1\}$)} \\
\end{array}
\]

\item A noncommutative cubic form $\sum_{ijk} T_{ijk} x_i x_j x_k$ is \emph{multilinear} if all nonzero terms $T_{ijk}$ have $i,j,k$ distinct (that is, $|\{i,j,k\}|=3$). The system of equations for $\{\pm 1\}$\textsc{-Monomial Equivalence of Adjective Noncommutative Cubic Forms} is the same as the above, with the restriction that $T$ and $T'$ both satisfy \textsc{Adjective} (e.\,g., multilinear, nonzero entries in $\{\pm 1\}$, etc.). 
\end{enumerate}
\label{defn:monomial_equivalence}
\end{definition}

\begin{theorem} \label{thm:XORtoMon}
There is a linear-size (1,3)-reduction from \textsc{Random 3-XOR} instances on $n$ variables with $m$ clauses, where $10^4 n \leq m \leq \binom{n}{3}/10^{12}$, to $\{\pm 1\}$\textsc{-Monomial Equivalence of $\{\pm 1\}$ Multilinear Noncommutative Cubic Forms}, over any ring $R$ of characteristic $\neq 2$.
\end{theorem}

The reduction is always a (1,3)-reduction, but we only show the resulting system of equations for $\{\pm 1\}$\textsc{-Monomial Equivalence of Noncommutative Cubic Forms} is unsatisfiable with probability $1-o(1)$ when the 3-XOR instance is chosen randomly with the parameters specified in the theorem. (It is possible that it is always unsatisfiable when the input 3-XOR instance is, but our proof does not answer this question.) 

\begin{proof}[Proof idea]
We build multilinear noncommutative cubic forms from the 3-XOR instance such that they are equivalent by a $\{\pm 1\}$ diagonal matrix iff the 3-XOR instance is satisfiable: an equation $x_i x_j x_k = \pm 1$ corresponds to setting $T_{ijk} = 1, T'_{ijk} = \pm 1$ in this construction. The noncommutative cubic forms are multilinear because the construction of the random 3XOR instance ensures that each XOR clause contains 3 distinct variables. In fact, the equations for $\{\pm 1\}$-diagonal equivalence of the correspondence noncommutative cubic forms will turn out to be identically the same as the equations for the 3-XOR instance. 

Next, for random instances chosen with the stated parameters, the 3-way arrays $T,T'$ are the adjacency hyper-matrices of a 3-uniform hypergraph that has no nontrivial automorphisms by \cite[Lemma~6.9]{DBLP:conf/soda/ODonnellWWZ14}; this is why we needed to restrict the parameter range for $m$ as we did. Because the hypergraphs have no nontrivial automorphisms, any monomial equivalence of the corresponding cubic forms must in fact be diagonal, thus letting us further reduce to $\{\pm 1\}$-monomial equivalence.
\end{proof}

\begin{proof}
We are given a system of 3-XOR equations, which we'll denote $x_{i_\ell} x_{j_\ell} x_{k_\ell} = s_\ell$ for $\ell=1,\dotsc,m$, where $i_\ell \leq j_\ell \leq k_\ell \in [n]$ are indices of variables and $s_{\ell} \in \{\pm 1\}$ for all $\ell$. It also includes the equations $x_i^2 = 1$ for all $i=1,\dotsc,n$.

\textbf{Step 1: Reduce from random 3-XOR to $\{\pm 1\}$-diagonal equivalence of noncommutative cubic forms.} From the above system of equations, we now construct two $n \times n \times n$ 3-way arrays $T,T'$. For the original equations $x_{i_\ell} x_{j_\ell} x_{k_\ell} = s_\ell$ ($\ell=1,\dotsc, m$),  and for any $a_\ell \in \{\pm 1\}$ of our choice (we may set all $a_\ell = 1$ if we wish, but this additional flexibility may be useful in other settings) we set
\[
T_{i_\ell, j_\ell, k_\ell} = a_\ell \text{ and } T'_{i_\ell, j_\ell, k_\ell} = s_\ell a_\ell.
\]
All other entries of $T$ and $T'$ are set to zero.

We start with a warmup lemma, to see that this part of the construction already has a desirable property. By a ``$\{\pm 1\}$ diagonal isomorphism'' of two non-commutative cubic forms, we mean a diagonal matrix $X$ whose diagonal entries are all one of $\pm 1$ such that $X$ gives an equivalence between $T,T'$.

\begin{lemma} \label{lem:warmup}
Notation as in the paragraph above. There is a bijection between the solutions to the 3-XOR instance and the $\{\pm 1\}$ diagonal isomorphisms of the noncommutative cubic forms defined by $T,T'$.
\end{lemma}

\begin{proof}
Suppose $\mathbf{x}$ is a solution to the 3-XOR instance. Let $X = \diag(x_1, \dotsc, x_n)$ be the diagonal matrix with $\mathbf{x}$ on the diagonal. We claim that $X$ is an equivalence between the noncommutative cubic forms represented by $T,T'$, or the same, that $(X,X,X)$ is an isomorphism of the tensors $T,T'$. Note that for any diagonal matrices $X,Y,Z$, we have $((X,Y,Z) \cdot T)_{ijk} = x_i y_j z_k T_{ijk}$. In particular, the action of diagonal matrices does not change which entries of $T$ are zero or nonzero, it merely scales the nonzero entries. Since $T,T'$ have the same support by construction, it is necessary and sufficient to handle the nonzero entries. By the construction above, there are precisely $m$ such nonzero entries, one for each cubic equation in the 3-XOR instance. For each $\ell=1,\dotsc,m$, we have
\begin{align*}
((X,X,X) \cdot T)_{i_\ell j_\ell k_\ell} & = x_{i_\ell} x_{j_\ell} x_{k_\ell} T_{i_\ell j_\ell k_\ell} \\
 & = s_\ell T_{i_\ell j_\ell k_\ell} \\
 & = T'_{i_\ell j_\ell k_\ell}.
\end{align*}

In the other direction, if $X=\diag(\mathbf{x})$ is a diagonal matrix whose diagonal entries are in $\{\pm 1\}$ giving an isomorphism of the noncommutative cubic forms, then we have
\[
x_{i_\ell} x_{j_\ell} x_{k_\ell} = T_{i_\ell j_\ell k_\ell} T'_{i_\ell j_\ell k_\ell} = s_\ell
\]
for $\ell=1,\dotsc,m$. (Here we have pulled $T_{i_\ell, j_\ell, k_\ell}$ across the equals sign because every term in the above equation is $\pm 1$.) This concludes the proof of the lemma.
\end{proof}

We thus consider the equations corresponding to $\{\pm 1\}$-diagonal equivalence of $T,T'$: there are $n$ variables $x_i$ ($i=1,\dotsc,n$). Let $X$ denote the diagonal matrix with $\mathbf{x}$ on the diagonal. Then the equations are
\begin{equation} \label{eq:diag}
X^2 = \Id \qquad (X,X,X) \cdot T = T'.
\end{equation}
By Lemma~\ref{lem:warmup}, we have that the original 3XOR instance is satisfiable iff (\ref{eq:diag}) is satisfiable. We claim furthermore that there is (1,3)-reduction from the 3XOR equations to this system of equations. In fact, as the proof of the preceding lemma shows, they are actually \emph{the same set of equations}! So there is nothing more to show.

\textbf{Step 2: Reduce from $\{\pm 1\}$-diagonal equivalence to $\{\pm 1\}$-monomial equivalence.} We claim that there is a $(1,3)$-reduction from (\ref{eq:diag}) to the the equations for $\{\pm 1\}$-monomial equivalence, see (\ref{defn:monomial_equivalence}). The variable substitution is given by
\[
a_{ij} = a'_{ij}  \mapsto \begin{cases}
0 & i \neq j \\
x_i & i = j.
\end{cases}
\]
Under this substitution:
\begin{itemize}
\item The equivalence condition $(A,A,A) \cdot T = T'$ becomes exactly the original equivalence condition $(X,X,X) \cdot T = T'$.

\item The invertibility equations $AA' = A'A = \Id$ become $XX = \Id$

\item The row and column equations both become $0 = 0$, since at least one of the two $a_{ij}$ variables occurring will not be on the diagonal, hence will become 0 after substitution.

\item The equation $a_{ij} (a_{ij}+1)(a_{ij}-1) = 0$ becomes $x (x^2-1)$ for the appropriate variable $x \in \mathbf{x}$. This is derivable from the original equation $x^2 -1$ by multiplication by $x$. 
\end{itemize}

Lastly, we show that the system of equations in Definition~\ref{defn:monomial_equivalence}(3) for $\{\pm 1\}$-monomial equivalence is satisfiable iff the original 3-XOR instance was. Since we showed above that that $\{\pm 1\}$-diagonal equivalence equations are satisfiable iff the original 3-XOR instance was, we show the equisolvability of (\ref{eq:diag}) and the equations of Definition~\ref{defn:monomial_equivalence}(3). 

Since diagonal matrices are monomial, any solution to (\ref{eq:diag}) is a solution to the equations of Definition~\ref{defn:monomial_equivalence}(3).

Conversely, suppose the equations of Definition~\ref{defn:monomial_equivalence}(3) are solvable. Then there is a $\{\pm 1\}$-monomial matrix $X$ given an equivalence between $T$ and $T'$; we may write $X=DP$ where $D$ is diagonal and $P$ is a permutation matrix. Now, as the original 3-XOR instance was chosen uniformly at random, the support of $T$ (the positions of its nonzero entries) is precisely a uniformly random 3-uniform hypergraph $G$. As $T,T'$ have the same support by construction, we find that $P$ must be an automorphism of $G$. But by \cite[Lemma~6.9]{DBLP:conf/soda/ODonnellWWZ14}, uniformly random such hypergraphs have no nontrivial automorphisms with probability $1-o(1)$. Thus $P=I$ and $X$ must in fact be diagonal, hence a solution to (\ref{eq:diag}).
\end{proof}

\begin{remark}
We may avoid the heavy hammer of \cite[Lemma~6.9]{DBLP:conf/soda/ODonnellWWZ14} by ``rigidifying'' (in the sense of removing automorphisms) the system of 3-XOR equations before constructing the 3-way arrays as follows. The construction corresponds to a standard graph-theoretic gadget for removing automorphisms. Add new variables $z$ and $y_{ij}$ for $i=1,\dotsc, n$ and $j=1,\dotsc,n+i$, as well as the equations $x_i y_{ij} z = 1$ for all $i,j$, as well as $y_{ij}^2 = 1$ and $z^2 = 1$. The downside of this construction is that it quadratically increases the number of variables, which would result in a further quadratic loss in our lower bounds on \textsc{Tensor Isomorphism}.
\end{remark}

\subsection{From $\{\pm 1\}$-monomial equivalence to (unrestricted) monomial equivalence}
\begin{theorem} \label{thm:pm1_to_monomial}
There is a linear-size (2, 6)-many-one reduction from 
\begin{center}
$\{\pm 1\}$\textsc{-Monomial Equivalence of $\{\pm 1\}$ Multilinear Noncommutative Cubic Forms}

to

\textsc{Monomial Equivalence of $\{\pm 1\}$ Noncommutative Cubic Forms},
\end{center}
over any ring $R$ of characteristic $\neq 2$ such that $\{\pm 1\}$ are the only square roots of $1$. 

Furthermore, the reduction $r$ has the property that, given any two $\{\pm 1\}$ multilinear noncommutative cubic forms $f,f'$, any monomial equivalence between $r(f)$ and $r(f')$ must have all its nonzero entries sixth roots of unity, and this can be derived by a degree-6 PC proof.
\end{theorem}

\begin{remark}
We note the difference between a reduction to $\sqrt[6]{1}$-\textsc{Monomial Equivalence} and a reduction to \textsc{Monomial Equivalence} with the property stated in the theorem. In the former case, the problem being reduced to only accepts $\sqrt[6]{1}$-monomial matrices as solutions (and then the goal of the reduction is to introduce gadgets to get this down to $\{\pm 1\}$). In the latter case, the problem being reduced to allows arbitrary monomial matrices as solutions, but the gadgets enforce that, on the reduced instances, any such monomial matrix must in fact have its nonzero entries being sixth roots of unity.
\end{remark}

\begin{proof}
Let $T$ be an $n \times n \times n$ 3-way array representing a multilinear noncommutative cubic form with all nonzero entries in $\pm 1$. We extend $T$ to $r(T)$ of size $2n \times 2n \times 2n$, by setting
\begin{align*}
r(T)_{ijk} & = T_{ijk} & i,j,k \in[n] \\
r(T)_{i,i,n+i} & = 1 & i \in [n] \\
r(T)_{n+i,n+i,n+i} & = 1 & i \in [n]
\end{align*}
and all other entries of $r(T)$ set to zero. 

\textbf{Many-one reduction.} We first show that the map $(T,T') \mapsto (r(T), r(T'))$ is a many-one reduction. Suppose $T,T'$ are $\{\pm 1\}$-monomially equivalent by a matrix $X$, where $X = DP$ with $D = \diag(x_1, \dotsc, x_n)$ a diagonal matrix with $x_i \in \{\pm 1\}$ for all $i$, and $P$ is a permutation matrix. Let $\pi$ denote the permutation corresponding to $P$; that is, $P_{i,\pi(i)} = 1$ for all $i \in [n]$. Then we claim the $2n \times 2n$ matrix $X \oplus P = \begin{bmatrix} X & 0 \\ 0 & P \end{bmatrix}$ is a monomial equivalence of $r(T)$ with $r(T')$. Since $X \oplus P$ is block-diagonal, the upper-left $X$ certainly sends the upper-left $n \times n \times n$ sub-array of $r(T)$ (which is just $T$) to that of $r(T')$ (which is just $T'$). So the only thing to check is what happens to the positions at indices greater than $n$.

Let $X' = X \oplus P$. We have
\begin{align*}
((X',X',X') \cdot r(T))_{i,i,n+i} & = r(T)_{\pi(i), \pi(i), n+\pi(i)} (X'_{i,\pi(i)})^2 X'_{n+i,n+\pi(i)} \\
 & = r(T)_{\pi(i), \pi(i), n+\pi(i)} (X_{i,\pi(i)})^2 P_{i,\pi(i)} \\
 & = 1 = r(T')_{i,i,n+i}.
\end{align*}
Similarly, we have:
\begin{align*}
((X',X',X') \cdot r(T))_{n+i,n+i,n+i} & = r(T)_{n+\pi(i), n+\pi(i), n+\pi(i)} P_{i,\pi(i)}^3= 1 = r(T')_{n+i, n+i,n+i}
\end{align*}
Because $X'$ is monomial, it is easy to see that the zeros of $r(T)$ are sent to zeros of $r(T')$. Thus $X'$ is a monomial equivalence of $r(T)$ with $r(T')$. 

Conversely, suppose $r(T)$ and $r(T')$ are equivalent by a monomial matrix $Y=DP$, with $D$ diagonal and $P$ a permutation matrix corresponding to permutation $\pi \in S_{2n}$. We will show that this implies that $T$ and $T'$ are equivalent by a $\{\pm 1\}$ monomial matrix. Since $T$ is multilinear, we have $T_{i,i,i}=r(T)_{i,i,i}=0$. Since $r(T)_{n+j,n+j,n+j}=1$ for all $j \in [n]$, the permutation $\pi$ cannot send any element $> n$ to any element $\leq n$. Thus $P$ is block-diagonal, say $P = \begin{bmatrix} P_1 & 0_n \\ 0_n & P_2 \end{bmatrix}$. Let $\pi_1$ (resp., $\pi_2$) be the permutation of $[n]$ corresponding to $P_1$ (resp., $P_2$). 

Next, we claim $P_1 = P_2$. By considering the positions at indices $(i,i,n+i)$, we have:
\begin{align*}
((P,P,P) \cdot r(T))_{i,i,n+i} = r(T)_{\pi_1(i), \pi_1(i), n+\pi_2(i)}
\end{align*}
But the latter is equal to the corresponding position in $r(T')$, which is $1$ iff $\pi_1(i)=\pi_2(i)$. Since this holds for all $i$, we have $\pi_1 = \pi_2$, and thus $P_1 = P_2$. 

Finally, we \emph{do not} claim that the diagonal entries $y_i$ themselves must be in $\pm 1$. Rather, we will show that they are all sixth roots of unity. Then cubing them will yield a new $n \times n$ matrix $D'$ all of whose diagonal entries are $\pm 1$ such that $D' P_1$ is a $\pm 1$-monomial equivalence of $T$ with $T'$. 

From the positions $(n+i,n+i,n+i)$, we have
\begin{align*}
1 & = r(T')_{n+\pi_1(i),n+\pi_1(i),n+\pi_1(i)} \\
& = ((Y,Y,Y) \cdot r(T))_{n+i,n+i,n+i}\\
& = y_{n+i}^3.
\end{align*}
But then, considering the positions $(i,i,n+i)$, we similarly get that $y_i^2 y_{n+i} = 1$. Cubing the latter equation, we get $y_i^6 y_{n+i}^3 = 1$. But as we already have $y_{n+i}^3 = 1$, this gives us $y_i^6 = 1$ by a degree-6 PC proof, as claimed in the ``furthermore.''  

Now we use the fact that $T,T'$ have all entries in $\{0,\pm 1\}$. Thus, each nonzero entry of $r(T)$ in the front-upper-left block (corresponding to $T$) gives us an equation of the form $y_i y_j y_k T_{ijk} = T'_{\pi_1(i), \pi_1(j), \pi_1(k)}$. Since the nonzero entries of $T,T'$ are $\pm 1$, this is thus an equation of the form $y_i y_j y_k = \pm 1$. If we cube both sides of this equation, we get $y_i^3 y_j^3 y_k^3 = \pm 1$. But since we established above that $y_i^6=1$ for all $i$, we have that $y_i^3 \in \{\pm 1\}$ for all $i$. Thus, defining $x_i := y_i^3$ for $i=1,\dotsc, n$, we have $x_i \in \{\pm 1\}$ and letting $D' = \diag(x_1, \dotsc, x_n)$, we have $D' P_1$ is a $\{\pm 1\}$-monomial equivalence from $T$ to $T'$.

\textbf{Low-degree PC reduction.} We claim that the system of equations for $\{\pm 1\}$ monomial equivalence of $T$ and $T'$ is (2,6)-reducible to the system of equations for monomial equivalence of $r(T)$ and $r(T')$. Let $X,X'$ be the $n \times n$ variable matrices for the equations for for $\{\pm 1\}$-monomial equivalence of the original tensors $T$ and $T'$, and let $Y,Y'$ be the $2n \times 2n$ matrices for the equations for monomial equivalence of $r(T),r(T')$. The PC reduction is defined by the following substitution:
\begin{align*}
y_{ij} & \mapsto x_{ij} & i,j \in [n] \\
y_{n+i, n+j} & \mapsto x_{ij}^2 & i,j \in [n] \\
y_{i,n+j}, y_{n+i,j} & \mapsto 0 & i,j \in [n],
\end{align*}
and similarly for the $y'$ variables being substituted by the $x'$ variables. That is, we have
\[
Y \mapsto \begin{bmatrix} X & 0_n \\ 0_n & X \circ X \end{bmatrix}
\qquad
Y' \mapsto \begin{bmatrix} X' & 0_n \\ 0_n & X' \circ X' \end{bmatrix},
\]
where $X \circ X$ denotes the entrywise (aka Hadamard) product with itself, that is $(X \circ X)_{ij} = x_{ij}^2$. The reason to use $X \circ X$ here is that if $X$ is $\{\pm 1\}$-valued and monomial, then $X \circ X$ is the permutation matrix with the same support as $X$; that is, this substitution is essentially the same as the one used in the proof above for the many-one reduction.

Now, taking advantage of the block structure in the substitution above and the block structure in $r(T), r(T')$, let us see what our equations become after substitution, and how to derive them from the equations for $T,T'$. This will complete the proof.
\begin{enumerate}
\item The set of equations $(Y, Y, Y) \cdot r(T) = r(T')$ becomes the set of equations $(X,X,X) \cdot T = T'$ (by examining the front-upper-left corner), as well as the equations 
\[
\sum_{i,j,k \in [2n]} y_{ii'} y_{jj'} y_{k,k'} r(T)_{ijk} = \begin{cases}
1 & i'=j'=k'-n \text{ or } i'=j'=k' > n \\
0 & \text{otherwise.}
\end{cases}
\]
We deal with the three cases ($i'=j'=k'-n$, $i'=j'=k'>n$, or neither of these) separately.
\begin{enumerate}
\item Suppose $i'=j'=k'-n$. In this case, $y_{ii'}$ is only nonzero for $i \in [n]$, and similarly for $y_{jj'}$, while $y_{kk'}$ is only nonzero for $k > n$. Thus the substituted equation becomes
\[
\sum_{i,j,k \in [n]} y_{ii'} y_{ji'} y_{n+k,n+i'} r(T)_{i,j,n+k} = \sum_{i,j,k \in [n]} x_{ii'} x_{ji'} x_{k,i'}^2 r(T)_{i,j,n+k} = 1
\]
Now, the only positions in $r(T)$ of the form $(i,j,n+k)$ with $i,j,k \in [n]$ that are nonzero are those of the form $(i,i,n+i)$, so the preceding equation simplifies further to
\[
\sum_{i \in [n]} x_{ii'} x_{ii'} x_{ii'}^2 = 1
\]
i.e., 
\begin{equation} \label{eq:monomial1}
\sum_{i \in [n]} x_{ii'}^4 = 1.
\end{equation}

We will now show how to derive (\ref{eq:monomial1}) from the equations for $\{\pm 1\}$-monomial equivalence of for $T,T'$ (Definition~\ref{def:equations for equivalence}). 
From the $\{0, \pm 1\}$ equation in Definition~\ref{def:equations for equivalence}(3), if we multiply by $x_{ii'}$, we get 
\begin{equation} \label{eq:monomial3}
x_{ii'}^2 (x_{ii'}^2 - 1),
\end{equation}
i.e., the usual Boolean equation but for $x_{ii'}^2$ rather than $x_{ii'}$ itself. 
Next, from $x_{i i'} x_{i'' i'}$ with $i \neq i''$, we may square this to get
\begin{equation} \label{eq:monomial4}
x_{ii'}^2 x_{i'' i'}^2.
\end{equation}
and we similarly get $(x'_{i'i})^2 (x'_{i' i''})^2$ when $i \neq i''$. 

Lastly, from the equation $XX' = \Id$ and multiplying by $\sum_{i \in [n]} x_{ii'} x_{i'i}' + 1$, we obtain

\begin{equation} \label{eq:monomial5}
(\sum_{i \in [n]} x_{ii'} x_{i'i}' + 1) (\sum_{i \in [n]} x_{ii'} x'_{i'i} - 1) = \sum_{i \in [n]} x_{ii'}^2 x_{i' i}^2 + \sum_{i, j \in [n] \\ i \neq j} x_{ii'} x'_{i'i} x_{ji'} x'_{i'j} - 1 = \sum_{i \in [n]} x_{ii'}^2 x_{i' i}^2 - 1,
\end{equation}

where we observed that from the axioms that $x_{ii'} x_{ji'} = 0$ for $i \neq j$ we may derive in degree 4 that the middle term $\sum_{i, j \in [n] \\ i \neq j} x_{ii'} x_{ji'} x'_{i'i} x'_{i'j} = 0$.

Now, equations (\ref{eq:monomial3})--(\ref{eq:monomial5}) are a degree-2 substitution instance of the equations in Lemma~\ref{lem:dotproduct} with $c=2, d=1$. Thus, by Lemma~\ref{lem:dotproduct}, we can derive (\ref{eq:monomial1}) from these in degree 6.
 
\item Suppose $i'=j'=k' > n$. In this case, the substitution makes all of $y_{ii'}, y_{jj'}, y_{kk'}$ equal to zero unless $i,j,k > n$. Thus we may write the equation, after substitution, as
\begin{align*}
\sum_{i,j,k \in [n]} y_{n+i,i'} y_{n+j,i} y_{n+k,i} r(T)_{n+i, n+j, n+k} & = \sum_{i,j,k \in [n]} x_{i,i'-n}^2 x_{j,i'-n}^2 x_{k,i'-n}^2 r(T)_{n+i, n+j, n+k} \\
 & = r(T')_{i',i',i'} = 1.
\end{align*}
However, because the only entries $r(T)_{n+i,n+j,n+k}$ that are nonzero are those in which $i=j=k$, this simplifies further to:
\[
\sum_{i \in [n]} x_{i,i'-n}^6 = 1.
\]
This is a degree-2 substitution instance of Lemma~\ref{lem:dotproduct} with $c=3,d=1$, so it can be derived in degree 6 from the equations derived in part (a).

\item Suppose neither of the previous two cases hold.  The derivation will depend on which of $i',j',k'$ lie in $[n]$ versus $\{n+1, \dotsc, 2n\}$. 
\begin{enumerate}
\item When all are in $[n]$, we are in the front-upper-left corner of the tensor, and we exactly get the equations $(X,X,X) \cdot T = T'$. 

\item When all three of $i',j',k'$ are $> n$, the only nonzero entries of $r(T)$ are of the form $r(T)_{n+i,n+i,n+i}$, so the equation becomes
\[
\sum_{i \in [n]} x_{i,i'-n}^2 x_{i,j'-n}^2 x_{i,k'-n}^2 = 0.
\]
Since we have assumed $|\{i',j',k'\}| > 1$, there are at least two distinct indices among them, and thus each term in this sum is a multiple of one of our $x_{ij}x_{ij'}$ axioms with $j \neq j'$. 

\item Next, suppose instead that $i',j' \in [n], k' > n$. In this case, the only nonzero entries of $Y$ after substitution are those with $i,j \in [n], k > n$. Thus the equation becomes
\[
\sum_{i,j,k \in [n]} x_{ii'} x_{jj'} x_{k,k'-n}^2 r(T)_{i,j,n+k} = 0
\]
However, the only nonzero entries of $r(T)$ in which the first two coordinates are $\leq n$ and the third is $n+k$ are those of the form $i=j=k$, so the preceding becomes
\[
\sum_{i \in [n]} x_{ii'} x_{ij'} x_{ik'-n}^2 = 0.
\]
Since we do not have $i'=j'=k'-n$ (as that was covered in a previous case), at least two of the column indices differ, and thus each term of this sum is divisible by one of the axioms of the form $x_{ij} x_{ij'}$ with $j \neq j'$. 

\item In all other cases, the corresponding entries of $r(T)$ are all zero, so the equation reduces to $0=0$.
\end{enumerate}
\end{enumerate}

\item The equations $YY' = Y'Y = \Id$ become $XX' = X'X = \Id$ and $(X \circ X)(X' \circ X') = (X' \circ X')(X \circ X) = \Id$. The first of these is one of our original equations, so it remains to derive the second. We show how to derive $(X \circ X)(X' \circ X') = \Id$; the other is similar. For clarity, let us write it out using indices:
\begin{eqnarray} \label{eq:inverse}
\sum_j x_{ij}^2 (x'_{jk})^2 & - & \delta_{ik} = 0\qquad \forall i,k \in [n]
\end{eqnarray}
Starting from the equation $\sum_j x_{ij} x'_{jk} - \delta_{ik} = 0$, we multiply by $\sum_j x_{ij} x'_{jk}$, to get
\[
\sum_j x_{ij}^2 (x'_{jk})^2 + \sum_{j \neq j'}  x_{ij} x'_{jk} x_{ij'} x_{j'k} - \delta_{ik} \sum_j x_{ij} x'_{jk}.
\]
Note that every term in the middle summation here is divisible by some $x_{ij} x_{ij'}$ with $j \neq j'$, which is one of our equations, so we may cancel off those terms using those equations in degree 4. If $i \neq k$, then we are done. If $i=k$, then we add in our equation $\sum_j x_{ij} x'_{jk} - 1$ to get (\ref{eq:inverse}). 

\item The equations $y_{ij} y_{ij'} = 0$ for $j \neq j'$ become 0 after substitution unless $i,j,j'$ are either all in $[n]$ or all in $\{n+1,\dotsc,2n\}$. In the former case, the substituted equation is $x_{ij} x_{ij'} = 0$, which is already one of the original equations. In the latter case, the equation becomes $x_{ij}^2 x_{ij'}^2 = 0$; but this is easily derivable from $x_{ij} x_{ij'}$ by multiplying it by itself (degree 4). The equations saying there is at most one entry per column of $Y$ are derived from those for $X$ similarly.
\end{enumerate}
This covers all the equations for monomial equivalence of $r(T), r(T')$, and thus we are done.
\end{proof}

\begin{lemma} \label{lem:dotproduct}
For any integers $d \geq 1, c \geq 1$, from the equations
\[
x_i(x_i^d -1) (\forall i) \qquad x_i x_j (\forall i \neq j) \qquad \sum_{i=1}^n x_i y_i -1
\]
there is a degree-$\max\{d+2, cd\}$ PC derivation (over any ring $R$) of
\[
\sum_{i \in [n]} x_i^{cd} - 1
\]
\end{lemma}

Although in the proof above we only used the $d=1$ and $c=2,3$, we will later have occasion to use this lemma with larger values of $d$ and $c$, which is why we phrase it in this level of generality.

\begin{proof}
First we show it for $c=1$, then derive the general case from that. 

Let $S = \sum_{i \in [n]} x_i^{d}$, $D = \sum_{i \in [n]} x_i y_i$. Our first goal is to derive $S-1$. For each $i=1,\dotsc,n$, we can derive $x_i y_i (S-1)$ in degree $d+2$ as follows:
\begin{align*}
x_i y_i (S-1) & = x_i^{d+1} y_i + y_i \sum_{j \neq i} x_i x_j^d - x_i y_i \\
 & = y_i (x_i^{d+1} - x_i) + y_i \sum_{j \neq i} x_i x_j^d 
 = \underline{x_i (x_i^{d} - 1)} y_i  + y_i \sum_{j \neq i} \underline{x_i x_j} x_j^{d-1},
\end{align*}
where we have underlined the use of the axioms. 

Summing up the preceding for all $i$, we derive $DS - D$ in degree $d+2$. Finally, we multiply the starting equation $D-1$ by $S$ to get $SD-S$, also in degree $d+2$. Then we have
\[
(DS - D) - (SD - S) + (D-1) = S-1 = \sum_{i} x_i^d - 1,
\]
as desired.

For $c > 1$, we then sum the preceding with $\sum_{i \in [n]} (x_i^{(c-1)d-1} + x_i^{(c-2)d-1} + \dotsb + x_i^{d-1}) \underline{(x_i^{d+1} - x_i)} = \sum_{i \in [n]} x_i^{cd} - x_i^d$, which has degree $cd$.
\end{proof}

\subsection{From monomial equivalence to general equivalence of noncommutative cubic forms}
\begin{theorem} \label{thm:monomial to general equivalence}
There is a quadratic-size many-one reduction from
\begin{center}
\textsc{Monomial Equivalence of Noncommutative Cubic Forms}

to

\textsc{Equivalence of Noncommutative Cubic Forms},
\end{center}
over any field. 

If furthermore the input cubic forms $f,f'$ have the property that any monomial equivalence between them must have its nonzero scalars being $d$-th roots of unity, and the latter can be derived by PC in degree $d+1$, then the reduction above is a $(d,2d)$-many-one reduction.
\end{theorem}

\begin{proof}
Let $f$ be a noncommutative cubic form in variables $u_1, \dotsc, u_n$. Then $r(f)$ will be a new noncommutative cubic form, in $n + 2n(n+1)$ variables $u_1, \dotsc, u_n, v_{11}, v_{12}, \dotsc, v_{n,n+1}, w_{11}, w_{12}, \dotsc, w_{n,n+1}$, which is $r(f) = f + \sum_{i \in [n], j \in [n+1]} u_i v_{ij} w_{ij}$. In terms of the underlying three-way arrays, if we have $f = \sum_{i,j,k \in [n]} T_{ijk} u_i u_j u_k$, then we use $r(T)$ to denote the array underlying $r(f)$, which can be described as follows. The 3-way array $r(T)$ will have size $N \times N \times N$ where $N = n + 2n(n+1)$. Let $T_i$ denote the $i$-th frontal slice of $T_i$, that is, $T_i$ is the matrix such that $(T_i)_{jk} = T_{ijk}$. For $i=1,\dotsc,n$, the frontal slices of $r(T)$ will be defined as:
\[
r(T)_{i} = \left(\begin{array}{c;{2pt/2pt}cccccc;{2pt/2pt}cccccc}
T_i & & & & & &           & &   \\ \hdashline[2pt/2pt]
      & 0_{n+1} & & & & & & 0_{n+1} \\
      & & 0_{n+1} & & & & & & 0_{n+1} \\
      & & & \ddots & &    & & & & \ddots\\
      & & & & 0_{n+1} & & & & & & I_{n+1} \\
      & & & & & \ddots    & & & & & & \ddots \\
      & & & & & & 0_{n+1} &  & & & & & 0_{n+1} \\ \hdashline[2pt/2pt] 
      & 0_{n+1} & & & & & & 0_{n+1} \\
      & & 0_{n+1} & & & & & & 0_{n+1} \\
      & & & \ddots & &    & & & & \ddots\\
      & & & & 0_{n+1} & & & & & & 0_{n+1} \\
      & & & & & \ddots    & & & & & & \ddots \\
      & & & & & & 0_{n+1} &  & & & & & 0_{n+1}
\end{array}\right),
\]
where the $I_{n+1}$ occurs in the $i$-th $(n+1) \times (n+1)$ block of its region. That is, the lower-right $2n(n+1) \times 2n(n+1)$ sub-matrix is the Kronecker product $E_{i,n+i} \otimes I_{n+1}$, where $E_{i,n+i}$ is the $2n \times 2n$ matrix with a 1 in position $(i,n+i)$ and zeros everywhere else. For the slices $i=n+1, \dotsc, N$ we will have $r(T)_i = 0$. 

Our main claim is that the map $(T,T') \mapsto (r(T), r(T)')$ is the reduction claimed in the theorem.

\textbf{Many-one reduction.} Suppose $X \cdot f = f'$ with $X$ monomial. Write $X = PD$ with $D$ diagonal and $P$ a permutation matrix corresponding to the permutation $\pi \in S_n$. Then we claim that 
\[
Y = X \oplus ((P D^{-1}) \otimes I_{n+1}) \oplus (P \otimes I_{n+1})
\]
is an equivalence between $r(f)$ and $r(f')$, where here we assume our variables are ordered as above. For we have 

\begin{align*}
Y \cdot r(f) & = \sum_{ijk \in [n]} T_{ijk} (Y u_i) (Y u_j) (Y u_k) + \sum_{i \in [n],j \in [n+1]} (Y u_i) (Y v_{ij}) (Y w_{ij}) \\
 & = \sum_{ijk \in [n]} T_{ijk} (X u_i) (X u_j) (X u_k) + \sum_{i \in [n], j \in [n+1]} (X u_i) (PD^{-1} v_{ij})(P w_{ij}) \\
 & = X \cdot f  + \sum_{i \in [n],j \in [n+1]} D_{ii} u_{\pi(i)} (D_{ii}^{-1} v_{\pi(i), j}) w_{\pi(i),j} \\
 & = f' + \sum_{i \in [n], j \in [n+1]} u_{\pi(i)} v_{\pi(i),j} w_{\pi(i),j} \\
 & = r(f').
\end{align*}
The final inequality here follows from the fact that $\pi$ is a permutation, so the final sum includes all terms of the form $u_{i} v_{ij} w_{ij}$, just listed in a different order than originally.

Conversely, suppose $Y \cdot r(f) = r(f')$ for an arbitrary invertible $N \times N$ matrix $Y$. To find an equivalence between $f$ and $f'$, here we find it more useful to take the viewpoint of the 3-way arrays $r(T)$ and $r(T')$ corresponding to $r(f)$ and $r(f')$, respectively. 

The way $Y$ acts on the 3-way array $r(T)$ is to first take linear combinations of the frontal slices, say by replacing the $i$-th slice with $\sum_{j \in [N]} Y_{ij} r(T)_j$ (corresponding to the action of $Y$ on the third variable in each monomial), and then to take each slice $S$ and replace it by $Y S Y^t$ (the left multiplication corresponds to the action on the first variable in each monomial, and the right multiplication corresponds to the action on the second variable in each monomial). As this latter transformation preserves the rank of each slice, we will use the ranks of linear combinations of the slices to reason about properties of $Y$. 

\textbf{Claim 1:} $Y$ is a block-diagonal sum of an $n \times n$ matrix $X$ and a $2n(n+1) \times 2n(n+1)$ matrix. 

\begin{proof}[Proof of claim 1.]
First we show that $Y$ is block-triangular. To see this, note that since the last $2n(n+1)$ slices are zero, the action of $Y$ by taking linear combinations of slices cannot send any of the first $n$ slices to the last $2n(n+1)$ slices. That is, $Y$ has the form $Y = \begin{bmatrix} X & Z \\ 0 & W \end{bmatrix}$ where $X$ is $n \times n$ and $W$ is $2n(n+1) \times 2n(n+1)$. It remains to show that $Z$ must be zero.

Since $Y$ is block-diagonal and invertible, we have that $X$ and $W$ are each invertible.

Let $R$ be the tensor gotten from $r(T)$ by having $Y$ act by taking linear combinations of the slices. That is, the $i$-th frontal slices of $R$ is $R_i = \sum_{j \in [N]} Y_{ij} r(T)_j$. Since each slice $r(T)_i$ has its support in the upper-left $n \times n$ sub-matrix and the middle-right $n(n+1) \times n(n+1)$ sub-matrix, so does each slice $R_i$. Write 
\[
R_i = \begin{bmatrix} R_{i}^{(1,1)} & 0 & 0 \\ 0 & 0 & R_i^{(2,2)} \\ 0 & 0_{n(n+1)} & 0 \end{bmatrix},
\]
 where $R_i^{(1,1)}$ is $n \times n$ and $R_i^{(2,2)}$ is $n(n+1) \times n(n+1)$. 

Now consider the action of $Y$ that sends $R_i$ to $Y R_i Y^t = r(T')_i$. We now break up $Y$ further into blocks commensurate with how we wrote $R_i$ above; write 
\[
Y = \begin{bmatrix} X & A & B \\ 0 & C & D \\ 0 & E & F \end{bmatrix}
\qquad 
Z = \begin{bmatrix} A & B \end{bmatrix}
\qquad 
W = \begin{bmatrix} C & D \\ E & F \end{bmatrix}, 
\]
where $A,B$ are $n \times n(n+1)$, and $C,D,E,F$ are each $n(n+1) \times n(n+1)$. Then we have:
\begin{align*}
YR_i Y^t & =  \begin{bmatrix} X & A & B \\ 0 & C & D \\ 0 & E & F \end{bmatrix} 
\begin{bmatrix} R_{i}^{(1,1)} & 0 & 0 \\ 0 & 0 & R_i^{(2,2)} \\ 0 & 0_{n(n+1)} & 0 \end{bmatrix}
\begin{bmatrix} X^t & 0 & 0 \\ A^t & C^t & E^t \\ B^t & D^t & F^t \end{bmatrix}  \\
 & = \begin{bmatrix} X R_i^{(1,1)} & 0 & A R_i^{(2,2)} \\ 0 & 0 & C R_i^{(2,2)} \\ 0 & 0 & E R_i^{(2,2)}\end{bmatrix} 
 \begin{bmatrix} X^t & 0 & 0 \\ A^t & C^t & E^t \\ B^t & D^t & F^t \end{bmatrix} \\
 & = \begin{bmatrix}  XR_i^{(1,1)} X^t + AR_i^{(2,2)} B^t & AR_i^{(2,2)} D^t & AR_i^{(2,2)} F^t \\ 
 C R_i^{(2,2)} B^t & * & * \\
 E R_i^{(2,2)} B^t & * & * \\ 
 \end{bmatrix},
 \end{align*}
 where we have put $*$'s in positions we won't need in the argument. 
 
Next, since each of the first $n$ slices of $r(T')$ must be of this form, and those slices have zeros in each block except the $(1,1)$ and $(2,3)$ blocks, by considering the blocks $(1,2), (1,3), (2,1), (3,1)$ we must have
\[
A R_i^{(2,2)} D^t = 0 \qquad A R_i^{(2,2)} F^t = 0 \qquad C R_i^{(2,2)} B^t = 0 \qquad E R_i^{(2,2)} B^t = 0.
\]
For reasons that will become clear below, we combine these into the two equations
\[
A R_i^{(2,2)} \begin{bmatrix} D^t & F^t \end{bmatrix} = 0 \qquad \begin{bmatrix} C \\ E \end{bmatrix} R_i^{(2,2)} B^t = 0.
\]
Note that the $n(n+1) \times 2n(n+1)$ matrices $\begin{bmatrix} D^t & F^t \end{bmatrix}$ and $\begin{bmatrix} C^t & E^t \end{bmatrix}$ must both be full rank, since otherwise $W = \begin{bmatrix} C & D \\ E & F \end{bmatrix}$ would not be invertible.

The sum of the (2,3) blocks (of size $n(n+1) \times n(n+1)$) of the first $n$ slices of $r(T)$ is precisely the identity matrix $I_{n(n+1)}$. Thus, the linear span of these blocks contains an invertible matrix in it. Since $Y$ is invertible, that linear span is the same as the linear span of the blocks $\{R_i^{(2,2)} : i \in [n]\}$. Thus the latter contains a full-rank matrix, say $\sum_{i=1}^n \alpha_i R_{i}^{(2,2)}$. But since we have $A R_i^{(2,2)} \begin{bmatrix} D^t & F^t \end{bmatrix} = 0$ for all $i$, we may left multiply by $A$ and right-multiply by $\begin{bmatrix} D^t & F^t \end{bmatrix}$ to get $A \left(\sum_{i=1}^n \alpha_i R_i^{(2,2} \right) \begin{bmatrix} D^t & F^t \end{bmatrix} = \sum_{i=1}^n \alpha_i A R_i^{(2,2)} \begin{bmatrix} D^t & F^t \end{bmatrix} = 0$. But now we have that $\sum \alpha_i R_i^{(2,2)}$ is invertible, and $\begin{bmatrix} D^t & F^t \end{bmatrix}$ has full rank $n(n+1)$, so their product also has full rank $n(n+1)$. But then we have that $A$ times a full rank matrix is equal to $0$, hence $A$ must be zero. The same argument, \emph{mutatis mutandis}, using the equation $\begin{bmatrix} C \\ E \end{bmatrix} R_i^{(2,2)} B^t = 0$, gives us that $B = 0$. Hence $Y$ is block-diagonal as claimed.
\end{proof}

Next, we use properties of the ranks of the slices coming from the $I_{n+1}$ gadgets to show that $X$ must in fact be monomial.

\textbf{Claim 2:} $Y = \begin{bmatrix} X & 0 \\ 0 & W \end{bmatrix}$ where $X$ is monomial.

\begin{proof}
In both $r(T)$ and $r(T')$, any linear combination consisting of $k$ of the first $n$ slices (with nonzero coefficients) has rank in the range $[k(n+1), k(n+1) + n]$, for any $k=0,\dotsc,n$. The lower bound can be seen by noting that any such linear combination is block-diagonal with $k$ copies of $I_{n+1}$ on the block diagonal of the $(2,3)$ block. The upper bound comes from the fact that these are the only nonzero blocks in the lower-right $2n(n+1) \times 2n(n+1)$ sub-matrix, and the only other nonzero entries are in the $n \times n$ upper-left sub-matrix, which has rank at most $n$ because of its size.

Using notation from the proof of the preceding claim, since $Y R_i Y^t = r(T')_i$, and the latter has rank in the range $[n+1, 2n+1]$, $R_i$ must also have rank in the same range. But this is only possible if $R_i$ is a linear combination of precisely one of the first $n$ slices of $r(T)$. Thus, $X$ is monomial.
\end{proof}

From claim 2, we thus have that there is a permutation $\pi \in S_n$ and nonzero scalars $d_1, \dotsc, d_n$ such that $R_i = d_i r(T)_{\pi(i)}$ for all $i=1,\dotsc,n$, where $X = DP$ with $D$ the diagonal matrix with diagonal entries $d_i$ and $P$ the permutation matrix corresponding to $\pi$. Finally, in the proof of claim 1, we saw that the upper-left block of $Y R_i Y^t$ was $X R_{i}^{(1,1)} X^t + A R_i^{(2,2)} B^t$, and then learned that $A = B = 0$. Putting these together, and recalling that the upper-left block of $r(T)_i$ is $T_i$, we thus get 
\[
(DP) d_i T_{\pi(i)} (DP)^t = T'_i
\]
for all $i$. In other words, $X$ is a monomial equivalence from $T$ to $T'$ (hence, from $f$ to $f'$). This completes the proof that the construction gives a many-one reduction.

\textbf{Low-degree PC reduction.} To prove the ``furthermore'', suppose that the pair of cubic forms $f,f'$ has the property that any monomial equivalence between them must have its nonzero entries being $d$-th roots of unity, for some $d \geq 1$, and that this can be derived---more specifically, the equations $y_{ij}^{d+1} - y_{ij}$ and similarly for $y'_{ij}$---in degree $d+1$.

Let $Y,Y'$ be the variable matrices for (general) equivalence of $r(f), r(f')$; let $X,X'$ be the variable matrices for monomial equivalence of $f,f'$. Consider the substitution
\begin{equation} \label{eq:mon_to_general_sub}
Y \mapsto \begin{bmatrix} X & 0 \\ 0 & X^{\circ (d-1)} \otimes I_{n+1} \\ 0 & 0 & X^{\circ d} \otimes I_{n+1} \end{bmatrix} \qquad Y' \mapsto \begin{bmatrix} X' & 0 \\ 0 & (X')^{\circ (d-1)} \otimes I_{n-1}  \\ 0 & 0 & (X')^{\circ d} \otimes I_{n+1} \end{bmatrix},
\end{equation}
where $X^{\circ (d-1)}$ denotes the $(d-1)$-fold Hadamard product $X \circ X \circ \dotsb \circ X$, namely, $(X^{\circ (d-1)})_{ij} = x_{ij}^{d-1}$. We will show that the equations for equivalence of $r(f), r(f')$, after this substitution, can be derived from the equations for monomial equivalence of $f,f'$ in low-degree PC. 

(Note that the substitutions above correspond precisely to the forward direction of the many-one reduction, in which $X \oplus (D^{-1} P \otimes I_{n+1} ) \oplus (P \otimes I_{n+1})$ served as an equivalence. For, once we have $x_{ij}^{d+1} - x_{ij}$, we have $X^{\circ (d-1)} = D^{d-1} P = D^{-1} P$, and $X^{\circ d} = D^d P = P$.)

Recall that these equations are $Y \cdot r(f) = r(f')$ and $YY' = Y'Y = \Id$. The latter equations are easier to handle so we begin with those. They become $X^{\circ c} (X')^{\circ c} = (X')^{\circ c} X^{\circ c} = \Id$ for $c \in \{1,d-1,d\}$. For $c=1$, these are some of our starting equations. For $c > 1$, this is similar to the argument in Theorem~\ref{thm:pm1_to_monomial} (see the argument around Equation~(\ref{eq:inverse})), iterated, resulting in a proof of degree $2c$ for any $c$---in this case, $2d$.

Now to the equation(s) $Y \cdot r(f) = r(f')$. After substitution, these become
\begin{equation} \label{eq:running out of names}
\sum_{i,j,k \in [n]} T_{ijk} (X u_i)(X u_j) (X u_k) + \sum_{i \in [n], j \in [n+1]} (X u_i) (X^{\circ (d-1)} v_{ij}) (X^{\circ d} w_{ij}) = \sum_{ijk} T'_{ijk} u_i u_j u_k + \sum_{ij} u_i v_{ij} w_{ij}.
\end{equation}
Focusing on the first summations on both sides of the equation, we see these are precisely the equations $X \cdot f = f'$. After subtracting these off, we now deal with the remaining terms.

We have
\begin{align*}
\sum_{ij} u_i v_{ij} w_{ij} & = \sum_{i \in [n], j \in [n+1]} (X u_i) (X^{\circ (d-1)} v_{ij}) (X^{\circ d} w_{ij}) \\
 & = \sum_{i \in [n], j \in [n+1]} \left(\sum_{k \in [n]} x_{k,i} u_k \right) \left(\sum_{\ell \in [n]} x_{\ell,i}^{d-1} v_{\ell,j} \right)\left(\sum_{h \in [n]} x_{h,i}^{d} w_{h,j}  \right) \\
 & = \sum_{k,\ell \in [n], j \in [n+1]} u_k v_{\ell,j} w_{\ell,j} \left(\sum_{i \in [n]} x_{k,i} x_{\ell,i}^{d-1} x_{\ell,i}^d \right) + \sum_{\substack{k,\ell,h \in [n], j \in [n+1] \\ \ell \neq h}} u_k v_{\ell,j} w_{\ell',j} \left(\sum_{i \in [n]} x_{k,i} x_{\ell,i}^{d-1} x_{h,i}^d \right)
\end{align*}
This becomes the system of equations
\[
\begin{array}{rclr}
\delta_{k,\ell} & = & \sum_{i \in [n]} x_{k,i} x_{\ell,i}^{d-1} x_{\ell,i}^d & \qquad (\forall k,\ell \in [n]) \\
0 & = & \sum_{i \in [n]} x_{k,i} x_{\ell,i}^{d-1} x_{h,i}^{d} & \qquad (\forall k,\ell,h \in [n], \ell \neq h).
\end{array}
\]
(Note that technically we should quantify over all $j \in [n+1]$, but $j$ plays no role in these equations---it just serves to repeat the same equation $n+1$ times. This corresponds to the fact that the lower-right part of our matrices have the form $* \otimes I_{n+1}$.)

When $k \neq \ell$, every term in the first equation is a degree-$2d$ multiple of the monomial axiom $x_{k,i} x_{\ell,i}$. Similarly, every term in the second set of equations is a degree-$2d$ multiple of the monomial axiom $x_{\ell,i} x_{h,i}$. Thus all that remains is the first equation when $k=\ell$, namely, $1 = \sum_{i \in [n]} x_{k,i} x_{k,i}^{d-1} x_{k,i}^d$. This is derived in Lemma~\ref{lem:dotproduct}, with $c=2$ in degree $2d$ (since $d > 1$, we have $\max\{2d,d+2\} = 2d$). This completes the proof that we have a $(d,2d)$-reduction.
\end{proof}

\begin{remark}
There is a slightly simpler and smaller many-one reduction, namely $f \mapsto f + \sum_{i \in [n], j \in [n+1]} u_i v_{ij}^2$. However, in using that reduction, the witness for the forward direction becomes $X \oplus (D^{-1/2} P \otimes I_{n+1})$. This square root introduces a square into the equations that made it difficult to show that it was also a PC reduction. The reduction above fixes this issue.
\end{remark}

\subsection{From cubic forms to tensors}
Our reductions here are those from Futorny--Grochow--Sergeichuk \cite[Cor.~3.4 and Thm.~2.1]{FGS}. The many-one property follows from the results there. We prove that each of these reductions is in fact also a low-degree PC reduction between the corresponding polynomial solvability problems. They reduce first to a problem we call \textsc{Block Tensor Isomorphism}, and then from there to \textsc{Tensor Isomorphism}, so we begin by introducing the former problem and its associated equations.

\begin{definition}[{see Futorny--Grochow--Sergeichuk \cite{FGS}}]
A block $n \times m \times p$ 3-way array is a 3-way array together with a partition of its index sets $\{1,\dotsc,n\} = \{1, \dotsc, n_1\} \sqcup \{n_1 + 1, n_1 + 2, \dotsc, n_1+n_2\} \sqcup \dotsb \sqcup \{\sum_{i=1}^{N-1} n_i + 1, \dotsc, n\}$, and similarly for the other two directions. Two block 3-way arrays are said to be \emph{conformally partitioned} if they have the same size and the same partitions of their index sets. Two conformally partitioned 3-way arrays $T,T'$ with block sizes as above are \emph{block-isomorphic} (called ``block-equivalent'' in \cite{FGS}) if there exist invertible matrices $S_{11}, \dotsc, S_{1,N}, S_{21}, \dotsc, S_{2M}, S_{31}, \dotsc, S_{3P}$, where $S_{1,I}$ is of size $n_I \times n_I$, $S_{2,J}$ is of size $m_J \times m_J$, and $S_{2,K}$ is of size $p_K \times p_K$, such that the block-diagonal matrices give an isomorphism of tensors:
\[
(S_{11} \oplus S_{12} \oplus \dotsb \oplus S_{1N}, S_{21} \oplus \dotsb \oplus S_{2M}, S_{31} \oplus \dotsb \oplus S_{3P}) \cdot T = T'.
\]

Given two block 3-way arrays $T,T'$ as above, the equations for \textsc{Block Tensor Isomorphism} are as folllows. There are $2(\sum_{I \in [N]} n_i + \sum_{J \in [M]} m_J + \sum_{K \in [P]} p_K)$ variables arranged into $2(N + M + P)$ square matrices $X_{I},X'_{I}$ (of size $n_I \times n_I$), $Y_{J}, Y'_J$ (of size $m_J \times m_J$), and $Z_K, Z'_K$ (of size $p_K \times p_K$). Then the equations are:
\[
(X_1 \oplus \dotsb \oplus X_N, Y_1 \oplus \dotsb \oplus Y_M, Z_1 \oplus \dotsb \oplus Z_P) \cdot T = T'
\]
\[
X_I X'_I = X'_I X_I = \Id (\forall I \in [N]) \qquad Y_J Y'_J = Y'_J Y_J = \Id (\forall J \in [M]) \qquad Z_K Z'_K = Z'_K Z_K = \Id (\forall K \in [P]) 
\]
\end{definition}

\begin{lemma} \label{lem:forms to blocks}
The many-one reduction from
\begin{center}
\textsc{Equivalence of Noncommutative Cubic Forms}

to

\textsc{Block Tensor Isomorphism}
\end{center}
in \cite[Cor.~3.4]{FGS} is in fact a linear-size (1,3)-many-one reduction.
\end{lemma}

\begin{proof}
Given a noncommutative cubic form $f$ in $n$ variables, $f = \sum_{i,j,k \in [n]} T_{ijk} u_i u_j u_k$, we recall the block tensor $r(T)$ from \cite[Cor.~3.4]{FGS}. It is partitioned into $2 \times 3 \times 3$ many blocks, with the rows being partitioned into $n,1$, the columns into $n,n,1$, and the depths also into $n,n,1$; thus its total size is $(n+1) \times (2n+1) \times (2n+1)$. Let $E_{ijk}$ denote the tensor of this size whose only nonzero entry is a 1 in position $(i,j,k)$. Then we define
\[
r(T) = T + \sum_{i \in [n]} \left(E_{i,n+i,2n+1} + E_{i,2n+1,n+i} + E_{n+1,i,n+i} + E_{n+1,n+i,i} \right) + E_{n+1, 2n+1, 2n+1}
\]
If you wanted to think of this as part of the tensor corresponding to a cubic form, that cubic form would have $n+1$ new variables $v_1, \dotsc, v_n, z$, and the form would be:
\[
r(f) := f + \sum_{i \in [n]} (u_i v_i z + u_i z v_i + z u_i v_i + z v_i u_i) + z^3.
\]
(This doesn't quite line up with the above description of a tensor, as the tensor corresponding to $r(f)$ would necessarily have all 3 side lengths the same, $2n+1$. However, there are $n$ of the $2n+1$ rows in that tensor that are entirely zero, namely, the rows corresponding to those monomials that begin with a $v_i$.)

The equations for block isomorphism of $r(T)$ and $r(T')$ have the following variable matrices $X,X'$ are $n \times n$, $x,x'$ are $1 \times 1$, $Y_1, Y_1', Y_2, Y_2'$ are $n \times n$, $y, y'$ are $1 \times 1$, $Z_1, Z_1', Z_2, Z_2'$ are $n \times n$, and $z,z'$ are $1 \times 1$. Let $U,U'$ be the $n \times n$ variable matrices for the equations for equivalence of the noncommutative cubic forms $f,f'$. We consider the following substitution:
\[
X,Y_1, Z_1, Y'_2, Z'_2 \mapsto U \qquad X', Y_1', Z_1', Y_2, Z_2  \mapsto U' \qquad x,x',y,y',z,z' \mapsto 1.
\]
Under this substitution, the equations for block isomorphism of $r(T),r(T')$ become
\begin{align*}
(U,U,U) \cdot T & + \sum_{i \in [n]} \left( (U,U',1) \cdot E_{i,n+i,2n+1} + (U, 1,U') \cdot E_{i,2n+1,n+i} \right. \\
& \left. + (1,U,U') \cdot E_{n+1,i,n+i} + (1, U', U) \cdot E_{n+1,n+i,i} + (1,1,1) \cdot E_{n+1, 2n+1, 2n+1} \right) \\
= & T' +  \sum_{i \in [n]} \left(E_{i,n+i,2n+1} + E_{i,2n+1,n+i} + E_{n+1,i,n+i} + E_{n+1,n+i,i} \right) + E_{n+1, 2n+1, 2n+1}
\end{align*}
Now, because each summand inside the big sum corresponds to an identity matrix in a block (e.g. $\sum_{i \in [n]} E_{i,n+i,2n+1}$ is an identity matrix in rows $\{1,\dotsc,n\}$, columns $\{n+1, \dotsc, 2n\}$, and depth $2n+1$), the above equations give us many instances of $UU' = \Id$ and $U'U = \Id$, which is one of our starting equations. We also get the equation $1=1$, and lastly, $(U, U, U) \cdot T = T'$, which is another one of our starting equations. Thus the equations we get here are in fact precisely the same as the equations we started with. As these are cubic equations and the substitutions were linear, it is a (1,3)-PC reduction.
\end{proof}

\begin{lemma} \label{lem:blocks to tensors}
When the number of blocks is $O(1)$, the many-one reduction from
\begin{center}
\textsc{Block Tensor Isomorphism}

to

\textsc{Tensor Isomorphism}
\end{center}
in \cite[Thm.~2.1]{FGS} is in fact a quadratic-size (1,3)-many-one reduction.
\end{lemma}

Note that the output of the reduction of Lemma~\ref{lem:forms to blocks} has $2 \times 3 \times 3$ many blocks, so the restriction to $O(1)$ many blocks in Lemma~\ref{lem:blocks to tensors} presents no obstacle to our goal.

\begin{proof}
The key is to show how to effectively remove the partition in one of the three directions; then that reduction can be applied three times in the three separate directions. Let $T,T'$ be block tensors of size $n \times m \times p$, with $N \times M \times P$ many blocks. The construction of \cite[Lem.~2.2]{FGS} shows how to construct from this a block tensor of quadratic size with $N \times M \times 1$ many blocks. We recall the construction here and show that it is a (1,3)-PC reduction.

Let $p_1, \dotsc, p_P$ denote the sizes of the parts of the partition in the third direction. Let $r = \min\{n,m\}+1$---this will govern the rank of the identity matrix gadgets we add. Let $s = \sum_{K=1}^P 2^{K-1} r$ and $t = \sum_{K=1}^P 2^{K-1} r p_K$. Then the output tensor will have size $(n + s) \times (m + t) \times p$. (Note that, since $P = O(1)$, we have that $s$ is linearly bounded in $n,m$ and $t$ is quadratic as a function of $n,m,p$.) Let $T_1, \dotsc, T_p$ be the frontal slices of $T$. The $i$-th slice of $r(T)$ will be as follows. Suppose $i$ is in the $K$-th block, and write $i = i_0 + \sum_{k=1}^K p_k$ with $1 \leq i_0 \leq p_{K+1}$. Write the slices $T_i$ as $T_i = \begin{bmatrix} A_i & B_i \\ C_i & D_i \end{bmatrix}$, where $A_i$ is $n_1 \times m_1$---representing the first part of the partition of $T$ into rows and columns, and $D_i$ represents all the other parts.
Then we construct:
\[
r(T)_i := \left[\begin{array}{ccc;{2pt/2pt}c;{2pt/2pt}ccccc;{2pt/2pt}c;{2pt/2pt}ccc;{2pt/2pt}c|c}
0 & \dotsb & 0 &            &   &            &                      &            &    &           &    &            &    & & \\ \hdashline[2pt/2pt]      
              & &    & \ddots &    &            &                     &            &    &           &    &            &    & & \\ \hdashline[2pt/2pt]      
              & &    &            & 0 & \dotsb & I_{2^{K-1}r} & \dotsb & 0 &           &    &            &    & & \\ \hdashline[2pt/2pt]      
              & &    &            &    &            &                    &           &    & \ddots &    &            &    & & \\ \hdashline[2pt/2pt]             
              & &    &            &    &            &                    &           &    &            & 0 & \dotsb & 0 & & \\ \hdashline[2pt/2pt]
              & &    &            &    &            &                    &           &    &            &    &            &    & A_i & B_i \\ \hline
              & &    &            &    &            &                    &           &    &            &    &            &    & C_i & D_i \\              
\end{array}\right],
\]
where the $I_{2^{K-1}r}$ is in the $i_0$-th position within the $K$-th block-row and block-column as indicated by the dashed lines. Here the dashed lines \emph{do not} represent additional parts of the partition, they are just for visual clarity. The solid lines indicate the first part of the new partition into rows and columns. The rows of $C_i$ and $D_i$ are partitioned into blocks the same as they were originally in $T_i$, and the columns of $B_i$ and $D_i$ are partitioned into parts in the same way as they were originally in $T_i$. That is, the entire big gadget in the upper-left gets prepended to the first parts of the row and column partitions. This is the many-one reduction.

Let $X_{1}, \dotsc, X_N$, $Y_1, \dotsc, Y_M$, and $Z$ be variable matrices (with associated primes matrices $X_1'$, etc.), with sizes as follows:
\begin{itemize}
\item $X_1$ has size $(s + n_1) \times (s + n_1)$

\item $X_I$ for $I \geq 2$ has size $n_I \times n_I$

\item $Y_1$ has size $(t + m_1) \times (t + m_1)$

\item $Y_J$ for $J \geq 2$ has size $m_J \times m_J$

\item $Z$ has size $p \times p$.
\end{itemize}

We start from the equations for \textsc{Block Isomorphism} (but now where there is only one block in the third direction), namely
\[
X_I X_I' = X_I' X_I = \Id \qquad Y_J Y_J' = Y_J' Y_J = \Id \qquad ZZ' = Z'Z = \Id 
\]
and 
\[
(X_1 \oplus \dotsb \oplus X_N, Y_1 \oplus \dotsb \oplus Y_N, Z) \cdot r(T) = r(T').
\]
We make the following substitution (with the same substitutions, \emph{mutatis mutandis}, for the primed variables):
\begin{itemize}
\item $X_1 \mapsto I_s \oplus \hat{X}_1$, where $\hat{X}_1$ is a matrix of variables of size $n_1 \times n_1$.

\item For $I \geq 2$, $X_I$ maps to itself.

\item $Y_1 \mapsto I_t \oplus \hat{Y}_1$, where $\hat{Y}_1$ is a matrix of variables of size $m_1 \times m_1$.

\item For $J \geq 2$, $Y_J$ maps to itself.

\item $Z$ maps to a block matrix $Z_1 \oplus \dotsb \oplus Z_P$, where for each $K \in [P]$, we have $Z_K$ is a $p_K \times p_K$ matrix of variables.
\end{itemize}

Under these substitutions, the equations for \textsc{Block Isomorphism} of $r(T), r(T')$ become precisely the original equations for \textsc{Block Isomorphism} of $T,T'$, together with equations of the form $I_s E_i I_t = E_i$, where $E_i$ is the $s \times t$ gadget matrix in the upper-left in the $i$-th slice. Thus we get a $(1,3)$-reduction.

Finally, this is then repeated in the other two directions to reduce the number of blocks in all three directions to one, thus giving an instance of \textsc{Tensor Isomorphism}.
\end{proof}

\subsection{Putting it all together}
Finally, we combine all the above to prove Theorem~\ref{thm:lowerbound}.

\begin{proof}[Proof of Theorem~\ref{thm:lowerbound}]
Let $m = c n$ with $c \geq 10^4$. By Theorem~\ref{thm:random}, random 3XOR instances with clause density $c$ require PC degree $\Omega(n/c^2) = \Omega(n)$ (in our case) to refute. The number of instances that the random distribution assigns nonzero probability is $\binom{2\binom{n}{3}}{m} \sim \binom{n^3}{cn} \geq n^{3cn} / (cn)^{cn} = c^{2cn \log n - cn} \geq c^{\Omega(n \log n)}$.

By Theorem~\ref{thm:XORtoMon}, there is a (1,3)-many-one reduction from those instances to $\{\pm 1\}$\textsc{-Monomial Equivalence of $\{\pm 1\}$ Multilinear Noncommutative Cubic Forms}, where the number of variables in the cubic form is the same as the number of variables in the 3XOR instance. By Theorem~\ref{thm:pm1_to_monomial} there is then a (2,6)-many-one reduction to \textsc{Monomial Equivalence of $\{\pm 1\}$ Noncommutative Cubic Forms}, where the number of variables in the output cubic form is linear in the original number of variables, and such that the output forms have the property that any monomial equivalence between them has all its nonzero entries being 6-th roots of unity. This thus satisfies the hypothesis of Theorem~\ref{thm:monomial to general equivalence} with $d=6$, so there is a (6,12)-many-one reduction to \textsc{Equivalence of Noncommutative Cubic Forms}, where the output has a quadratic number of variables compared to the input. Finally, combining Lemmata~\ref{lem:forms to blocks} and \ref{lem:blocks to tensors}, we get a (1,3) reduction from \textsc{Equivalence of Noncommutative Cubic Forms} to \textsc{Tensor Isomorphism}, which further increases the size quadratically. In total, the size increases multiply, yielding a quartic size increase. The substitution degrees multiply and the derivation degrees we take the max, yielding a (12,12)-many-one reduction from Random 3XOR to \textsc{Tensor Isomorphism} on tensors of size $O(n^4) \times O(n^4) \times O(n^4)$. By Lemma~\ref{lem:PCred}, any PC refutation of these \textsc{Tensor Isomorphism} instances requires degree $\Omega(n)$.
\end{proof}

We note that our lower bound for tensor isomorphism also applies to the stronger Sum-of-Squares proof system. This is due to the fact that there is lower bound for random 3XOR in Sum-of-Squares, as shown by Grigoriev \cite{grigoriev2001linear} and independently by Schoenbeck \cite{schoenebeck2008linear}, which makes the dependence on the clause density explicit.

\begin{theorem}[{\cite[Theorem 12]{schoenebeck2008linear}}]
A random 3-XOR instance with clause density $\Delta = m/n = dn^\epsilon$, for all sufficiently large constants $d$, requires SoS degree $\Omega(n^{1 - \epsilon})$ to refute, with probability $1 - o(1)$. 
\label{thm:random_sos}
\end{theorem}

In particular, this is a linear $\Omega(n)$ lower bound in the case of constant clause density ($\epsilon = 0$), which matches the PC lower bound of Theorem \ref{thm:random}.

As we observe all of our reductions go through in Sum-of-Squares, since Sum-of-Squares simulates PC over the reals due to Berkholz \cite{berkholz2018relation}. Furthermore, this simulation preserves degrees of proofs up to a constant factor. 

\begin{theorem}[{\cite[Theorem 1.1]{berkholz2018relation}}]
If a system of polynomial equations $\mathcal{F}$ over the reals has a PC refutation of degree $d$ and size $s$, it also has a sum-of-squares refutation of degree $2d$ and size $poly(s)$.
\label{thm:sos_simulates_pc}
\end{theorem} 

Hence, by combining Theorems \ref{thm:random_sos}, \ref{thm:sos_simulates_pc} and the PC reduction used to prove \ref{thm:lowerbound}, we obtain the following lower bound for tensor isomorphism in Sum-of-Squares.

\begin{theorem}
Over the real numbers, there is a distribution on $n \times n \times n$ \textsc{Tensor Isomorphism} whose associated equations require SoS degree $\Omega(\sqrt[4]{n})$ to refute with probability 1 - o(1).
\label{thm:sos_lower_bound}   
\end{theorem}

%% file: GI.tex

\section{Lower bound on PC degree for Tensor Isomorphism from Graph Isomorphism} \label{sec:lowerboundGI}
\begin{definition} \label{def:GI}
Given two graphs $G,H$ with adjacency matrices $A,B$ (resp.), the equations for \textsc{Graph Isomorphism} (the same as those used by Berkholz \& Grohe \cite{BerkholzG15,BerkholzG16}) are as follows. Let $Z$ be an $n \times n$ matrix of variables $z_{ij}$ (where the intended interpretation is that $z_{ij}=1$ iff an isomorphism maps vertex $i \in V(G)$ to vertex $j \in V(H)$). We say that a partial map, which sends $(i,i') \mapsto (j,j')$ is a local isomorphism if (1) $i=i'$ iff $j=j'$ (it's a well-defined map) and (2) $(i,i') \in E(G) \Leftrightarrow (j,j') \in E(H)$. (One may also do \textsc{Colored Graph Isomorphism} and require that the colors match, $c(i) = c(j), c(i')=c(j')$.) Then the equations are:
\[
\begin{array}{rll}
z_{ij}^2 - z_{ij} & \forall i,j & \text{All variables $\{0,1\}$-valued} \\
1 - \sum_i z_{ij} & \forall j & \text{each $j \in V(H)$ is mapped to from exactly one vertex} \\
1 - \sum_j z_{ij} & \forall i & \text{each $i \in V(G)$ maps to exactly one vertex} \\
z_{ij} z_{i'j'} & & \text{Whenever $(i,i') \mapsto (j,j')$ is not a local isomorphism}.
\end{array}
\]
\end{definition}

In this section, we prove a lower bound on PC (and SoS) for \textsc{TI}, by reducing from \textsc{GI} and using the known lower bounds on \textsc{GI} \cite{BerkholzG15,BerkholzG16}. Specifically, we show

\begin{theorem}
Over any field, there are instances of \textsc{Tensor Isomorphism} of size $O(n) \times O(n) \times O(n)$ that require PC degree $\Omega(n)$ to refute. The same holds over the reals for SoS degree.
\end{theorem}

\begin{proof}
Berkholz and Grohe \cite{BerkholzG15,BerkholzG16} show the same statement for $n$-vertex graphs of bounded vertex degrees, with the same PC/SoS degree bound. In Proposition~\ref{prop:GIMon} we show that \textsc{GI} reduces to \textsc{Monomial Code Equivalence} by a (2,4)-many-one reduction that turns $n$-vertex, $m$-edge graphs into $m \times (3m+n)$ matrices. in Proposition~\ref{prop:MonTI} we show that \textsc{Monomial Code Equivalence} reduces to \textsc{TI} by a (2,4)-many-one reduction that turns $k \times N$ matrices into $(k + 2N) \times N \times (1 + 2N)$ tensors. By Lemma~\ref{lem:PCred}, this completes the proof.
\end{proof}

To reduce from \textsc{GI} to \textsc{TI} we use the following intermediate problem. A matrix is \emph{monomial} if it has exactly one nonzero entry in each row and column; equivalently, a monomial matrix is the product of a permutation matrix and an invertible diagonal matrix. 

\begin{definition}
\textsc{Monomial Code Equivalence} is the problem: given two $k \times n$ matrices $C,C'$, do there exist matrices $X,Y$ such that $XCY^T = C'$ where $X$ is invertible and $Y$ is invertible and monomial? Given two such matrices $C,C'$, the equations for \textsc{Monomial Code Equivalence} are as follows. There are $2(k^2 + n^2)$ variables arranged into matrices $X,X'$ (of size $k \times k$) and $Y,Y'$ (of size $n \times n$). The equations are
\[
XCY^T = C' \qquad XX' = X'X =\Id \qquad YY' = Y'Y = \Id
\]
and
\begin{align*}
& y_{ij}y_{ij'} (\forall i \forall j \neq j') \qquad y_{ij} y_{i'j} (\forall i \neq i', \forall j) \\
& y'_{ij}y'_{ij'} (\forall i \forall j \neq j') \qquad y'_{ij} y'_{i'j} (\forall i \neq i', \forall j) 
\end{align*}
(Note: there are no equations forcing the variables to take on values in $\{0,1\}$.)
\end{definition}

\begin{proposition} \label{prop:GIMon}
The reduction of Petrank \& Roth \cite{PR} from \textsc{Graph Isomorphism} to \textsc{Linear Code Equivalence} over $\F_2$ in fact gives a (2,4)-many-one reduction from \textsc{Graph Isomorphism} to \textsc{Monomial Code Equivalence} (sic!) over any field.
\end{proposition}

\begin{proof}
The reduction of Petrank \& Roth is as follows: given a simple undirected graph $G$ with $n$ vertices and $m$ edges, let $D(G)$ be its $m \times n$ incidence matrix: $D_{e,v} = 1$ iff $v \in e$ and is 0 otherwise, and let $M(G)$ be the $m \times (3m + n)$ matrix 
\[
M(G) = \left[
\begin{array}{c|c|c|c}
I_m & I_m & I_m & D(G)
\end{array}
\right].
\]

\textbf{Many-one reduction.} It was previously shown (over $\F_2$ in \cite{PR} and over arbitrary fields in \cite[Lem.~II.4]{GrochowLie}) that this gives a many-one reduction to \textsc{Permutational Code Equivalence}. Here we observe that the same reduction also gives a reduction to \textsc{Monomial Code Equivalence}. Thus, all that remains to show is that if $M(G)$ and $M(H)$ are monomially equivalent, then $G$ must be isomorphic to $H$. 

In fact, what was shown in \cite{PR} (over arbitrary fields in \cite{GrochowLie}) is that, up to permutation and scaling of the rows, $M(G)$ is the unique generator matrix of its code satisfying the following properties: (1) $M(G)$ is $m \times (3m+n)$, (2) each row has Hamming weight $\leq 5$, (3) any linear combination that includes two or more rows with nonzero coefficients has Hamming weight $\geq 6$. 

Now, suppose $(X,Y)$ is a monomial equivalence of the codes $M(G), M(H)$. Then the rowspans of $M(G) Y^T$ and $M(H)$ are the same. Since $Y$ is monomial, if we consider just the supports of the rows of $M(G) Y^T$, up to re-ordering the rows, by the preceding paragraph, those supports must be the same as the supports of the rows of $M(H)$. Thus $X$ must also be monomial. Say $X = DP$ and $Y=EQ$ where $D,E$ are diagonal and $P,Q$ are permutation matrices. Then $P M(G) Q^T$ has the same support as $X M(G) Y^T = M(H)$, and since $P$ and $Q$ are permutation matrices and $M(G)$ and $M(H)$ have all entries in $\{0,1\}$, we must have $P M(G) Q^T = M(H)$. Thus $M(G)$ and $M(H)$ are in fact equivalent by a \emph{permutation} matrix (in place of the monomial matrix $Y$). Thus, by the fact that $(G,H) \mapsto (M(G), M(H))$ was a reduction to \textsc{Permutational Code Equivalence}, we conclude that $G \cong H$.

\textbf{Low-degree PC reduction.} Let $X,X',Y,Y'$ be the variable matrices in the equations for \textsc{Monomial Code Equivalence} of $M(G), M(H)$, and let $Z$ be the variable matrix in the equations for \textsc{Graph Isomorphism} of $G,H$. Let $n = |V(G)|, m = |E(G)|$; so, $X,X'$ are of size $m$, $Y,Y'$ are of size $3m+n$, and $Z$ is of size $n$. 

Let $Z^{(2)}$ denote the $\binom{n}{2} \times \binom{n}{2}$ matrix whose $(\{i,i'\}, \{j,j'\})$ entry is $z_{ij} z_{i'j'} + z_{ij'} z_{i'j}$. The idea is that if $Z$ is a map on the vertices, then $Z^{(2)}$ is the corresponding map on the edges; the two terms come from the fact that the edge $\{i,i'\}$ can be mapped to the edge $\{j,j'\}$ either by $(i,i') \mapsto (j,j')$ or by $(i,i') \mapsto (j',j)$. Note that, since $Z$ is a permutation matrix, at most one of these terms is nonzero, and thus $Z^{(2)}$ is also a $\{0,1\}$-matrix (in fact, a permutation matrix). Let $Z^{(2)}_E$ denote the $|E| \times |E|$ submatrix of $Z^{(2)}$ all of whose row indices are $\{i,i'\} \in E(G)$ and all of whose column indices are $\{j,j'\} \in E(H)$. Note also that $(Z^{(2)}_E)^T = (Z^T)^{(2)}_E$, so we use these notations interchangeably for convenience.

Now consider the following substitution:
\[
\begin{array}{ll}
X \mapsto (Z^{(2)}_E)^T & Y \mapsto (Z^T)^{(2)}_E \oplus (Z^T)^{(2)}_E \oplus (Z^T)^{(2)}_E \oplus (Z^T)  \\
X' \mapsto Z^{(2)}_E & Y' \mapsto Z^{(2)}_E \oplus Z^{(2)}_E \oplus Z^{(2)}_E \oplus Z \\
\end{array}
\]
After making these substitutions in the equations for \textsc{Monomial Code Equivalence} of $M(G), M(H)$, we get the equations
\begin{equation} \label{eq:monGI}
(Z^{(2)}_E)^T Z^{(2)}_E = Z^{(2)}_E (Z^{(2)}_E)^T = \Id_m \qquad (Z^{(2)}_E)^T D(G) Z = D(H) \qquad ZZ^T = Z^T Z = \Id_n
\end{equation}
along with equations saying that $Z$ and $Z^{(2)}_E$ are monomial.

We now show how to derive these equations in low-degree PC from the \textsc{GI} equations. 

The monomial equations for $Z$ are part of the \textsc{GI} equations, so there is nothing to do for those.

The monomial equations for $Z^{(2)}_E$ are of the form $(z_{ij} z_{i'j'} + z_{ij'} z_{i'j})(z_{k\ell} z_{k'\ell'} + z_{k\ell'} z_{k'\ell})$ where either (1) $\{i,i'\} = \{k,k'\}$ and $\{j,j'\} \neq \{\ell,\ell'\}$ or (2) vice versa. We expand out to get 
\[
z_{ij} z_{i'j'} z_{k\ell} z_{k'\ell'} + z_{ij} z_{i'j'} z_{k\ell'} z_{k'\ell} + z_{ij'} z_{i'j}z_{k\ell} z_{k'\ell'} + z_{ij'} z_{i'j} z_{k\ell'} z_{k'\ell}
\]
We show how to get this equation in case (1); case (2) follows similarly, \emph{mutatis mutandis}. In case (1), without loss of generality suppose that $i=k$, $i'=k'$, and $j \notin \{\ell,\ell'\}$. The first two terms are divisible by the \textsc{GI} equations $z_{ij} z_{i \ell}$ (since $i=k$ and $j \neq \ell$), the third term is divisible by $z_{i'j} z_{i' \ell'}$ (since $i'=k'$ and $j \neq \ell'$), and the last term is divisible by $z_{i'j} z_{i' \ell}$ similarly.

Next, the equations $ZZ^T=\Id_n$ are, expanded out, 
\[
\sum_j z_{ij} z_{ij} -1 (\forall i) \qquad \sum_j z_{ij} z_{kj} (\forall i \neq k).
\]
The first is gotten by linear combination from $1-\sum_j z_{ij}$ and the Boolean axioms $z_{ij}^2 - z_{ij}$. The second is a linear combination of the monomial axioms $z_{ij} z_{kj}$ (part of the local non-isomorphism axioms). Similarly for $Z^T Z = \Id$, using $1 - \sum_i z_{ij}$ instead.

Next, we expand out the equations $Z^{(2)}_E (Z^T)^{(2)}_E = \Id_m$, to get\footnote{We use the notation $\sum_{\{j,j'\} \in E(H)}$ to denote a sum in the index of summation takes on the value $e \in E(H)$ for each edge of $H$ exactly once. Because our edges are undirected, we only use such sums when the summand expression is itself invariant under swapping the roles of $j,j'$. If so desired, one could equivalently say $\sum_{j < j', \{j,j'\} \in E(H)}$.}
\[
\sum_{\{j,j'\} \in E(H)} (z_{ij} z_{i'j'} + z_{ij'} z_{i'j})(z_{kj} z_{k'j'} + z_{k'j} z_{kj'}) - \delta_{\{i,i'\}, \{k,k'\}} (\forall \{i,i'\}, \{k,k'\} \in E(G))
\]
Thus, for $\{i,i'\} \neq \{k,k'\}$, we need to derive
\[
\sum_{\{j,j'\} \in E(H)} \left( z_{ij} z_{i'j'} z_{kj} z_{k'j'}+ z_{ij'} z_{i'j} z_{kj} z_{k'j'} + z_{ij} z_{i'j'} z_{k'j} z_{kj'} + z_{ij'} z_{i'j} z_{k'j} z_{kj'} \right).
\]
Without loss of generality, suppose that $i \notin \{k,k'\}$. Then the first two terms of each summand are divisible by the \textsc{GI} equation $z_{ij} z_{kj}$, the third term is divisible by $z_{ij} z_{k'j}$, and the last term is divisible by $z_{ij'} z_{kj'}$. 
On the other hand, when $\{i,i'\} = \{k,k'\}$, we need to derive
\[
-1 + \sum_{\{j,j'\} \in E(H)} \left( z_{ij}^2 z_{i'j'}^2 + 2 z_{ij'} z_{i'j} z_{ij} z_{i'j'} + z_{ij'}^2 z_{i'j}^2 \right).
\]
The middle terms of each summand are divisible by the \textsc{GI} equations $z_{ij'} z_{ij}$. For the first and third terms, we can use the Boolean axioms to remove the squares, and thus we are left to derive
\begin{equation} \label{eq:Z2}
-1 + \sum_{\{j,j'\} \in E(H)} \left(z_{ij} z_{i'j'} + z_{ij'} z_{i'j} \right)
\end{equation}
We derive this from the \textsc{GI} equations as follows. Consider $(\sum_{j} z_{ij} -1)(\sum_{j'} z_{i'j'} -1) + (\sum_{j} z_{ij} -1) + (\sum_{j'} z_{i'j'} -1)$ and break up the resulting sum according to whether $j = j'$, $\{j,j'\} \in E(H)$ or $\{j,j'\} \notin E(H)$. Then we get
\[
\sum_j z_{ij} z_{i'j}  + \sum_{j,j' : \{j,j'\} \in E(H)} z_{ij} z_{i'j'} + \sum_{j \neq j' \{j,j'\} \notin E(H)} z_{ij} z_{i'j'} - 1
\]
Every summand in the first sum is a monomial axiom since $i \neq i'$. Every summand in the third sum is a local non-isomorphism axiom, since $\{i,i'\} \in E(G)$ but $\{j,j'\} \notin E(H)$. Note that every edge $\{j,j'\}$ of $E(H)$ is represented twice in the middle sum: once as $(j,j')$ and once as $(j',j)$. Thus, the above simplifies to 
\[
\sum_{\{j,j'\} \in E(H)} (z_{ij} z_{i'j'} + z_{ij'} z_{i'j}) -1,
\]
which is what we sought to derive. The derivation of $(Z^{(2)}_E)^T Z^{(2)}_E = \Id$ is similar.

Finally, we show how to derive the equation $(Z^{(2)}_E)^T D(G) Z = D(H)$ from the equations $ZA(G) = A(H) Z$, where $A(G)$ denotes the adjacency matrix of $G$, with $A(G)_{ii'} = 1$ iff $\{i,i'\} \in E(G)$. Writing out the equations in indices, we need to derive
\[
\sum_{\{i,i'\} \in E(G),k \in V(G)} \left(Z^{(2)}_{E}\right)_{\{i,i'\},\{j,j'\}} D(G)_{\{i,i'\}, k} z_{k \ell} = D(H)_{\{j,j'\},\ell} (\forall \ell \in V(H), \forall \{j,j'\} \in E(H))
\]
Using the fact that $D(G)_{\{i,i'\},k} = \delta_{ik} + \delta_{i'k}$ and the definition of $Z^{(2)}$, this is the same as
\[
\sum_{\{i,i'\} \in E(G),k \in V(G)} \left(z_{ij} z_{i'j'} + z_{ij'} z_{i'j} \right) (\delta_{ik} + \delta_{i'k}) z_{k \ell} = \delta_{j\ell} + \delta_{j'\ell} (\forall \ell \in V(H), \forall \{j,j'\} \in E(H))
\]
Thus we need to derive:
\[
\sum_{\{i,i'\} \in E(G)} \left(z_{ij} z_{i'j'} + z_{ij'} z_{i'j} \right) (z_{i \ell}  + z_{i' \ell})= \begin{cases}
1 & \ell \in \{j,j'\} \\
0 & \text{otherwise}.
\end{cases}
\]
Expanding out the summand, we find the four terms
\[
z_{ij} z_{i'j'} z_{i\ell} + z_{ij} z_{i'j'} z_{i' \ell} + z_{ij'} z_{i'j} z_{i\ell} + z_{ij'} z_{i'j} z_{i'\ell}.
\]
When $\ell \notin \{j,j'\}$, each of these terms is divisible by one of the monomial (local non-isomorphism) axioms, respectively: $z_{ij} z_{i\ell}$, $z_{i'j'} z_{i'\ell}$, $z_{ij'} z_{i\ell}$, and $z_{i'j} z_{i'\ell}$.

Finally, when $\ell \in \{j,j'\}$, without loss of generality suppose that $\ell = j$. Then the only terms that are not divisible by the monomial axioms as above are $z_{ij}^2 z_{i'j'} + z_{ij'} z_{i'j}^2$. Using the Boolean axioms we can easily convert each such summand to $z_{ij} z_{i'j'} + z_{ij'} z_{i'j}$. The derivation of the sum of these over all $\{i,i'\} \in E(G)$ is analogous, \emph{mutatis mutandis}, to the derivation of (\ref{eq:Z2}) above. This completes the proof.
\end{proof}

\begin{proposition} \label{prop:MonTI}
The many-one reduction from \textsc{Monomial Code Equivalence} to \textsc{Tensor Isomorphism} from Grochow \& Qiao \cite{GQ2} is in fact a (2, 4)-many-one reduction.
\end{proposition}

\begin{proof}
We recall the reduction, then prove that it is a low-degree PC reduction. Let $M$ be a $k \times n$ matrix. We build a 3-tensor of size $(k+2n) \times n \times (1 + 2n)$ as follows. The first frontal slice is $\begin{bmatrix} M \\ 0_{2n \times n} \end{bmatrix}$. The remaining $2n$ slices all have just a single nonzero entry, which serve to place a $2 \times 2$ identity matrix ``behind and perpendicular'' to $M$, one $2 \times 2$ matrix in each column. Let us index these slices by $[n] \times 2$. Then the $(i,b)$ slice has a $1$ in entry $(2(i-1) + b, i)$, for all $i \in [n], b \in [2]$. Let us call this tensor $r(M)$. Then the reduction maps $M,M'$ to $r(M), r(M')$.

Let $X,X',Y,Y',Z,Z'$ be the variable matrices for the \textsc{TI} equations for $r(M), r(M')$, and let $A,B,A',B'$ be the variable matrices for \textsc{Monomial Code Equivalence} of $M,M'$ (that is, $AMB^T = M'$, $A$ is invertible, $B$ is monomial and invertible). Consider the substitution:
\[
X \mapsto A \oplus (B \otimes I_2) \qquad Y \mapsto B \qquad Z \mapsto 1 \oplus (B' \circ B') \otimes I_2
\]
\[
X' \mapsto A' \oplus (B' \otimes I_2) \qquad Y' \mapsto B' \qquad Z' \mapsto 1 \oplus (B \circ B) \otimes I_2.
\]
As before, $B \circ B$ denotes the Hadamard or entry-wise product. Let us see what the \textsc{TI} equations become under this substitution. We get
\[
AMB^T = M' \qquad AA' = A'A = \Id \qquad BB' = B'B = \Id \qquad (B' \circ B')(B \circ B) = (B \circ B)(B' \circ B') = \Id
\]
Indeed, notice that the effect of the $B \otimes I_2$ in $X$ and the $B$ in $Y$ is that the row and column locations of the $2 \times 2$ matrix gadgets get permuted in the same way, and the gadget get multiplied by the \emph{square} of the nonzero entries of $B$. These are then multiplied by the $B' \circ B'$ in $Z$.

Now, we derive these equations from the equations for \textsc{Monomial Code Equivalence}. The first three are already present in the equations for \textsc{Monomial Code Equivalence}. The last one we expand out, to see that we need to derive:
\[
\sum_{j} b_{ij}^2 (b'_{jk})^2 = \delta_{ik} (\forall i,k)
\]
Now, for $i \neq k$, we may take the equation $\sum_j b_{ij} b'_{jk}$ and square it, to derive
\[
\sum_{j \neq j'} b_{ij} b'_{jk} + b_{ij'} b'_{j'k} + \sum_j b_{ij}^2 b_{j'k}^2.
\]
Each term in the first sum is divisible by one of the monomial axioms $b_{ij} b_{ij'}$ since $j \neq j'$, and the second sum is what we wanted to derive.

Finally, for $i = k$, we square the equation $\sum_j b_{ij} b'_{ji} - 1$ and add to it $2 \left(\sum_j b_{ij} b'_{ji} -1 \right)$. We then proceed to cancel terms with the monomial axioms as above, and end up with $\sum_j b_{ij}^2 (b'_{ji})^2 - 1$, as desired.
\end{proof}

%% file: PCforTensorIso.bbl
\newcommand{\etalchar}[1]{$^{#1}$}
\begin{thebibliography}{OWWZ14}

\bibitem[ABRW04]{ABRW}
Michael Alekhnovich, Eli Ben{-}Sasson, Alexander~A. Razborov, and Avi
  Wigderson.
\newblock Pseudorandom generators in propositional proof complexity.
\newblock {\em {SIAM} J. Comput.}, 34(1):67--88, 2004.
\newblock \href {https://doi.org/10.1137/S0097539701389944}
  {\path{doi:10.1137/S0097539701389944}}.

\bibitem[AM13]{AM}
Albert Atserias and Elitza~N. Maneva.
\newblock {Sherali}--{Adams} relaxations and indistinguishability in counting
  logics.
\newblock {\em {SIAM} J. Comput.}, 42(1):112--137, 2013.
\newblock \href {https://doi.org/10.1137/120867834}
  {\path{doi:10.1137/120867834}}.

\bibitem[AS05]{AS05}
Manindra Agrawal and Nitin Saxena.
\newblock Automorphisms of finite rings and applications to complexity of
  problems.
\newblock In {\em {STACS} 2005, 22nd Annual Symposium on Theoretical Aspects of
  Computer Science, Proceedings}, pages 1--17, 2005.
\newblock \href {https://doi.org/10.1007/978-3-540-31856-9_1}
  {\path{doi:10.1007/978-3-540-31856-9_1}}.

\bibitem[Bab16]{Babai}
L\'{a}szl\'{o} Babai.
\newblock Graph isomorphism in quasipolynomial time [extended abstract].
\newblock In {\em S{TOC}'16---{P}roceedings of the 48th {A}nnual {ACM} {SIGACT}
  {S}ymposium on {T}heory of {C}omputing}, pages 684--697. ACM, New York, 2016.
\newblock \href {https://doi.org/10.1145/2897518.2897542}
  {\path{doi:10.1145/2897518.2897542}}.

\bibitem[BCIP02]{BCIP}
Josh Buresh{-}Oppenheim, Matthew Clegg, Russell Impagliazzo, and Toniann
  Pitassi.
\newblock Homogenization and the polynomial calculus.
\newblock {\em Comput. Complex.}, 11(3-4):91--108, 2002.
\newblock \href {https://doi.org/10.1007/s00037-002-0171-6}
  {\path{doi:10.1007/s00037-002-0171-6}}.

\bibitem[Ber18]{berkholz2018relation}
Christoph Berkholz.
\newblock The relation between polynomial calculus, sherali-adams, and
  sum-of-squares proofs.
\newblock In {\em 35th Symposium on Theoretical Aspects of Computer Science
  (STACS 2018)}. Schloss Dagstuhl-Leibniz-Zentrum fuer Informatik, 2018.

\bibitem[BG15]{BerkholzG15}
Christoph Berkholz and Martin Grohe.
\newblock Limitations of algebraic approaches to graph isomorphism testing.
\newblock In Magn{\'{u}}s~M. Halld{\'{o}}rsson, Kazuo Iwama, Naoki Kobayashi,
  and Bettina Speckmann, editors, {\em Automata, Languages, and Programming -
  42nd International Colloquium, {ICALP} 2015, Kyoto, Japan, July 6-10, 2015,
  Proceedings, Part {I}}, volume 9134 of {\em Lecture Notes in Computer
  Science}, pages 155--166. Springer, 2015.

\bibitem[BG17]{BerkholzG16}
Christoph Berkholz and Martin Grohe.
\newblock Linear diophantine equations, group csps, and graph isomorphism.
\newblock In Philip~N. Klein, editor, {\em Proceedings of the Twenty-Eighth
  Annual {ACM-SIAM} Symposium on Discrete Algorithms, {SODA} 2017, Barcelona,
  Spain, Hotel Porta Fira, January 16-19}, pages 327--339. {SIAM}, 2017.
\newblock Preprint \href{http://arxiv.org/abs/1607.04287}{arXiv:1607.04287
  [cs.CC]}.
\newblock \href {https://doi.org/10.1137/1.9781611974782.21}
  {\path{doi:10.1137/1.9781611974782.21}}.

\bibitem[BGIP01]{DBLP:journals/jcss/BussGIP01}
Samuel~R. Buss, Dima Grigoriev, Russell Impagliazzo, and Toniann Pitassi.
\newblock Linear gaps between degrees for the polynomial calculus modulo
  distinct primes.
\newblock {\em J. Comput. Syst. Sci.}, 62(2):267--289, 2001.
\newblock \href {https://doi.org/10.1006/jcss.2000.1726}
  {\path{doi:10.1006/jcss.2000.1726}}.

\bibitem[BGL{\etalchar{+}}19]{BGLQW}
Peter~A. Brooksbank, Joshua~A. Grochow, Yinan Li, Youming Qiao, and James~B.
  Wilson.
\newblock Incorporating {Weisfeiler}--{Leman} into algorithms for group
  isomorphism.
\newblock \href{https://arxiv.org/abs/1905.02518}{arXiv:1905.02518 [cs.CC]},
  2019.

\bibitem[BI99]{BI}
Eli Ben{-}Sasson and Russell Impagliazzo.
\newblock Random {CNF}'s are hard for the {Polynomial} {Calculus}.
\newblock In {\em 40th Annual Symposium on Foundations of Computer Science,
  {FOCS} '99, 17-18 October, 1999, New York, NY, {USA}}, pages 415--421. {IEEE}
  Computer Society, 1999.
\newblock (Journal version in \emph{Comput. Complex.} 2010,
  doi:10.1007/s00037-010-0293-1).
\newblock \href {https://doi.org/10.1109/SFFCS.1999.814613}
  {\path{doi:10.1109/SFFCS.1999.814613}}.

\bibitem[Bre70]{brent}
R.~P. Brent.
\newblock Algorithms for matrix multiplication.
\newblock Stanford Computer Science Dept. Tech. Report STAN-CS-70-157,
  available online at
  \href{https://apps.dtic.mil/sti/pdfs/AD0705509.pdf}{https://apps.dtic.mil/sti/pdfs/AD0705509.pdf},
  1970.

\bibitem[BS20]{BS1}
Jendrik Brachter and Pascal Schweitzer.
\newblock On the {Weisfeiler}--{Leman} dimension of finite groups.
\newblock In Holger Hermanns, Lijun Zhang, Naoki Kobayashi, and Dale Miller,
  editors, {\em {LICS} '20: 35th Annual {ACM/IEEE} Symposium on Logic in
  Computer Science, Saarbr{\"{u}}cken, Germany, July 8-11, 2020}, pages
  287--300. {ACM}, 2020.
\newblock \href {https://doi.org/10.1145/3373718.3394786}
  {\path{doi:10.1145/3373718.3394786}}.

\bibitem[BS22]{BS2}
Jendrik Brachter and Pascal Schweitzer.
\newblock A systematic study of isomorphism invariants of finite groups via the
  {Weisfeiler}--{Leman} dimension.
\newblock In Shiri Chechik, Gonzalo Navarro, Eva Rotenberg, and Grzegorz
  Herman, editors, {\em 30th Annual European Symposium on Algorithms, {ESA}
  2022, September 5-9, 2022, Berlin/Potsdam, Germany}, volume 244 of {\em
  LIPIcs}, pages 27:1--27:14. Schloss Dagstuhl - Leibniz-Zentrum f{\"{u}}r
  Informatik, 2022.
\newblock \href {https://doi.org/10.4230/LIPIcs.ESA.2022.27}
  {\path{doi:10.4230/LIPIcs.ESA.2022.27}}.

\bibitem[CEI96]{CEI}
Matthew Clegg, Jeffery Edmonds, and Russell Impagliazzo.
\newblock Using the {G}roebner basis algorithm to find proofs of
  unsatisfiability.
\newblock In {\em Proceedings of the {T}wenty-eighth {A}nnual {ACM} {S}ymposium
  on the {T}heory of {C}omputing ({P}hiladelphia, {PA}, 1996)}, pages 174--183.
  ACM, New York, 1996.
\newblock \href {https://doi.org/10.1145/237814.237860}
  {\path{doi:10.1145/237814.237860}}.

\bibitem[CFI92]{CFI}
Jin-Yi Cai, Martin F\"{u}rer, and Neil Immerman.
\newblock An optimal lower bound on the number of variables for graph
  identification.
\newblock {\em Combinatorica}, 12(4):389--410, 1992.
\newblock \href {https://doi.org/10.1007/BF01305232}
  {\path{doi:10.1007/BF01305232}}.

\bibitem[CL22]{CL}
Nathaniel~A. Collins and Michael Levet.
\newblock Count-free {Weisfeiler}--{Leman} and group isomorphism.
\newblock \href{https://arxiv.org/abs/2212.11247}{arXiv:2212.11247 [cs.DS]},
  2022.

\bibitem[FGS19]{FGS}
Vyacheslav Futorny, Joshua~A. Grochow, and Vladimir~V. Sergeichuk.
\newblock Wildness for tensors.
\newblock {\em Linear Algebra Appl.}, 566:212--244, 2019.
\newblock \href {https://doi.org/10.1016/j.laa.2018.12.022}
  {\path{doi:10.1016/j.laa.2018.12.022}}.

\bibitem[FP06]{DBLP:conf/eurocrypt/FaugereP06}
Jean{-}Charles Faug{\`{e}}re and Ludovic Perret.
\newblock Polynomial equivalence problems: Algorithmic and theoretical aspects.
\newblock In Serge Vaudenay, editor, {\em Advances in Cryptology - {EUROCRYPT}
  2006, 25th Annual International Conference on the Theory and Applications of
  Cryptographic Techniques, St. Petersburg, Russia, May 28 - June 1, 2006,
  Proceedings}, volume 4004 of {\em Lecture Notes in Computer Science}, pages
  30--47. Springer, 2006.
\newblock \href {https://doi.org/10.1007/11761679\_3}
  {\path{doi:10.1007/11761679\_3}}.

\bibitem[GQ21a]{GQ2}
Joshua~A. Grochow and Youming Qiao.
\newblock On p-group isomorphism: Search-to-decision, counting-to-decision, and
  nilpotency class reductions via tensors.
\newblock In Valentine Kabanets, editor, {\em 36th Computational Complexity
  Conference, {CCC} 2021, July 20-23, 2021, Toronto, Ontario, Canada (Virtual
  Conference)}, volume 200 of {\em LIPIcs}, pages 16:1--16:38. Schloss Dagstuhl
  - Leibniz-Zentrum f{\"{u}}r Informatik, 2021.
\newblock \href {https://doi.org/10.4230/LIPIcs.CCC.2021.16}
  {\path{doi:10.4230/LIPIcs.CCC.2021.16}}.

\bibitem[GQ21b]{GQ}
Joshua~A. Grochow and Youming Qiao.
\newblock On the complexity of isomorphism problems for tensors, groups, and
  polynomials {I:} tensor isomorphism-completeness.
\newblock In James~R. Lee, editor, {\em 12th Innovations in Theoretical
  Computer Science Conference, {ITCS} 2021, January 6-8, 2021, Virtual
  Conference}, volume 185 of {\em LIPIcs}, pages 31:1--31:19. Schloss Dagstuhl
  - Leibniz-Zentrum f{\"{u}}r Informatik, 2021.
\newblock \href {https://doi.org/10.4230/LIPIcs.ITCS.2021.31}
  {\path{doi:10.4230/LIPIcs.ITCS.2021.31}}.

\bibitem[Gri81]{grigorievGI}
D.~Ju. Grigoriev.
\newblock Complexity of ``wild'' matrix problems and of the isomorphism of
  algebras and graphs.
\newblock {\em Zap. Nauchn. Sem. Leningrad. Otdel. Mat. Inst. Steklov. (LOMI)},
  105:10--17, 198, 1981.
\newblock Theoretical applications of the methods of mathematical logic, III.
\newblock \href {https://doi.org/10.1007/BF01084390}
  {\path{doi:10.1007/BF01084390}}.

\bibitem[Gri01]{grigoriev2001linear}
Dima Grigoriev.
\newblock Linear lower bound on degrees of positivstellensatz calculus proofs
  for the parity.
\newblock {\em Theoretical Computer Science}, 259(1-2):613--622, 2001.

\bibitem[Gri13]{grigoriev}
Dima Grigoriev.
\newblock Polynomial complexity of solving systems of few algebraic equations
  with small degrees.
\newblock In Vladimir~P. Gerdt, Wolfram Koepf, Ernst~W. Mayr, and Evgenii~V.
  Vorozhtsov, editors, {\em Computer Algebra in Scientific Computing - 15th
  International Workshop, {CASC} 2013, Berlin, Germany, September 9-13, 2013.
  Proceedings}, volume 8136 of {\em Lecture Notes in Computer Science}, pages
  136--139. Springer, 2013.
\newblock \href {https://doi.org/10.1007/978-3-319-02297-0\_11}
  {\path{doi:10.1007/978-3-319-02297-0\_11}}.

\bibitem[Gro12]{GrochowLie}
Joshua~A. Grochow.
\newblock Matrix {Lie} algebra isomorphism.
\newblock In {\em IEEE Conference on Computational Complexity (CCC12)}, pages
  203--213, 2012.
\newblock Also available as arXiv:1112.2012 [cs.CC] and ECCC Technical Report
  TR11-168.
\newblock \href {https://doi.org/10.1109/CCC.2012.34}
  {\path{doi:10.1109/CCC.2012.34}}.

\bibitem[Gro13]{GrochowSE}
Joshua~A. Grochow.
\newblock Answer to ``deciding bound on tensor rank for a fixed value''.
\newblock CSTheory StackExchange,
  \href{https://cstheory.stackexchange.com/a/19518/129}{https://cstheory.stackexchange.com/a/19518/129},
  2013.

\bibitem[Hel19]{helfgott}
Harald~Andr\'{e}s Helfgott.
\newblock Isomorphismes de graphes en temps quasi-polynomial [d'apr\`es {B}abai
  et {L}uks, {W}eisfeiler-{L}eman,\dots ].
\newblock {\em Ast\'{e}risque}, (407):Exp. No. 1125, 135--182, 2019.
\newblock S\'{e}minaire Bourbaki. Vol. 2016/2017. Expos\'{e}s 1120--1135.
  English translation with appendices by Jitendra Bajpai and Daniele Dona
  available at \href{https://arxiv.org/abs/1710.04574}{arXiv:17010.04574
  [math.GR]}.
\newblock \href {https://doi.org/10.24033/ast} {\path{doi:10.24033/ast}}.

\bibitem[HQ21]{HQ}
Xiaoyu He and Youming Qiao.
\newblock On the {B}aer-{L}ov\'{a}sz-{T}utte construction of groups from
  graphs: isomorphism types and homomorphism notions.
\newblock {\em European J. Combin.}, 98:Paper No. 103404, 12, 2021.
\newblock \href {https://doi.org/10.1016/j.ejc.2021.103404}
  {\path{doi:10.1016/j.ejc.2021.103404}}.

\bibitem[HT15]{HT}
Pavel Hrubes and Iddo Tzameret.
\newblock Short proofs for the determinant identities.
\newblock {\em {SIAM} J. Comput.}, 44(2):340--383, 2015.
\newblock \href {https://doi.org/10.1137/130917788}
  {\path{doi:10.1137/130917788}}.

\bibitem[JQSY19]{JQSY19}
Zhengfeng Ji, Youming Qiao, Fang Song, and Aaram Yun.
\newblock General linear group action on tensors: {A} candidate for
  post-quantum cryptography.
\newblock In Dennis Hofheinz and Alon Rosen, editors, {\em Theory of
  Cryptography - 17th International Conference, {TCC} 2019, Nuremberg, Germany,
  December 1-5, 2019, Proceedings, Part {I}}, volume 11891 of {\em Lecture
  Notes in Computer Science}, pages 251--281. Springer, 2019.
\newblock Preprint \href{https://arxiv.org/abs/1906.04330}{arXiv:1906.04330
  [cs.CR]}.
\newblock \href {https://doi.org/10.1007/978-3-030-36030-6\_11}
  {\path{doi:10.1007/978-3-030-36030-6\_11}}.

\bibitem[Lan12]{landsbergBook}
J.~M. Landsberg.
\newblock {\em Tensors: geometry and applications}, volume 128 of {\em Graduate
  Studies in Mathematics}.
\newblock American Mathematical Society, Providence, RI, 2012.
\newblock \href {https://doi.org/10.1090/gsm/128} {\path{doi:10.1090/gsm/128}}.

\bibitem[Las01]{lasserre}
Jean~B. Lasserre.
\newblock Global optimization with polynomials and the problem of moments.
\newblock {\em SIAM J. Optim.}, 11(3):796--817, 2000/01.
\newblock \href {https://doi.org/10.1137/S1052623400366802}
  {\path{doi:10.1137/S1052623400366802}}.

\bibitem[Lic85]{lickteigTypical}
Thomas Lickteig.
\newblock Typical tensorial rank.
\newblock {\em Linear Algebra Appl.}, 69:95--120, 1985.
\newblock \href {https://doi.org/10.1016/0024-3795(85)90070-9}
  {\path{doi:10.1016/0024-3795(85)90070-9}}.

\bibitem[LL89]{LL}
Thomas Lehmkuhl and Thomas Lickteig.
\newblock On the order of approximation in approximative triadic decompositions
  of tensors.
\newblock {\em Theoret. Comput. Sci.}, 66(1):1--14, 1989.
\newblock \href {https://doi.org/10.1016/0304-3975(89)90141-2}
  {\path{doi:10.1016/0304-3975(89)90141-2}}.

\bibitem[Luk93]{Luks}
Eugene~M. Luks.
\newblock Permutation groups and polynomial-time computation.
\newblock In {\em Groups and computation ({N}ew {B}runswick, {NJ}, 1991)},
  volume~11 of {\em DIMACS Ser. Discrete Math. Theoret. Comput. Sci.}, pages
  139--175. Amer. Math. Soc., Providence, RI, 1993.

\bibitem[McK81]{mckay}
Brendan~D. McKay.
\newblock Practical graph isomorphism.
\newblock {\em Congr. Numer.}, 30:45--87, 1981.

\bibitem[Miy96]{Miyazaki}
Takunari Miyazaki.
\newblock {Luks's} reduction of graph isomorphism to code equivalence.
\newblock Comment to E. W. Clark,
  \href{https://groups.google.com/forum/#!msg/sci.math.research/puZxGj9HXKI/CeyH2yyyNFUJ}{https://groups.google.com/forum/\#!msg/sci.math.research/puZxGj9HXKI/CeyH2yyyNFUJ},
  1996.

\bibitem[MP14]{MP}
Brendan~D. McKay and Adolfo Piperno.
\newblock Practical graph isomorphism, {II}.
\newblock {\em J. Symbolic Comput.}, 60:94--112, 2014.
\newblock \href {https://doi.org/10.1016/j.jsc.2013.09.003}
  {\path{doi:10.1016/j.jsc.2013.09.003}}.

\bibitem[NPOV15]{TNN}
Alexander Novikov, Dmitry Podoprikhin, Anton Osokin, and Dmitry Vetrov.
\newblock Tensorizing neural networks.
\newblock In {\em Proceedings of the 28th International Conference on Neural
  Information Processing Systems - Volume 1}, NIPS'15, pages 442--450. MIT
  Press, 2015.

\bibitem[OWWZ14]{DBLP:conf/soda/ODonnellWWZ14}
Ryan O'Donnell, John Wright, Chenggang Wu, and Yuan Zhou.
\newblock Hardness of robust graph isomorphism, lasserre gaps, and asymmetry of
  random graphs.
\newblock In Chandra Chekuri, editor, {\em Proceedings of the Twenty-Fifth
  Annual {ACM-SIAM} Symposium on Discrete Algorithms, {SODA} 2014, Portland,
  Oregon, USA, January 5-7, 2014}, pages 1659--1677. {SIAM}, 2014.
\newblock Preprint available as
  arXiv:\href{https://arxiv.org/abs/1401.2436}{1401.2436 [cs.CC]}.
\newblock \href {https://doi.org/10.1137/1.9781611973402.120}
  {\path{doi:10.1137/1.9781611973402.120}}.

\bibitem[Pat96]{Pat96}
Jacques Patarin.
\newblock Hidden fields equations {(HFE)} and isomorphisms of polynomials
  {(IP):} two new families of asymmetric algorithms.
\newblock In {\em Advances in Cryptology - {EUROCRYPT} '96, International
  Conference on the Theory and Application of Cryptographic Techniques,
  Saragossa, Spain, May 12-16, 1996, Proceeding}, pages 33--48, 1996.
\newblock \href {https://doi.org/10.1007/3-540-68339-9_4}
  {\path{doi:10.1007/3-540-68339-9_4}}.

\bibitem[PR97]{PR}
Erez Petrank and Ron~M. Roth.
\newblock Is code equivalence easy to decide?
\newblock {\em {IEEE} Trans. Inf. Theory}, 43(5):1602--1604, 1997.
\newblock \href {https://doi.org/10.1109/18.623157}
  {\path{doi:10.1109/18.623157}}.

\bibitem[Raz98]{DBLP:journals/cc/Razborov98}
Alexander~A. Razborov.
\newblock Lower bounds for the polynomial calculus.
\newblock {\em Comput. Complex.}, 7(4):291--324, 1998.
\newblock \href {https://doi.org/10.1007/s000370050013}
  {\path{doi:10.1007/s000370050013}}.

\bibitem[SC04]{SC}
Michael Soltys and Stephen Cook.
\newblock The proof complexity of linear algebra.
\newblock {\em Ann. Pure Appl. Logic}, 130(1-3):277--323, 2004.
\newblock \href {https://doi.org/10.1016/j.apal.2003.10.018}
  {\path{doi:10.1016/j.apal.2003.10.018}}.

\bibitem[Sch08]{schoenebeck2008linear}
Grant Schoenebeck.
\newblock Linear level lasserre lower bounds for certain k-csps.
\newblock In {\em 2008 49th Annual IEEE Symposium on Foundations of Computer
  Science}, pages 593--602. IEEE, 2008.

\bibitem[Sol01]{soltys}
Michael Soltys.
\newblock {\em The complexity of derivations of matrix identities}.
\newblock PhD thesis, University of Toronto, 2001.
\newblock Availalble on ECCC at
  \href{https://eccc.weizmann.ac.il/resources/pdf/soltys.pdf}{https://eccc.weizmann.ac.il/resources/pdf/soltys.pdf}.

\bibitem[Som99]{sombra}
Mart\'{\i}n Sombra.
\newblock A sparse effective {N}ullstellensatz.
\newblock {\em Adv. in Appl. Math.}, 22(2):271--295, 1999.
\newblock \href {https://doi.org/10.1006/aama.1998.0633}
  {\path{doi:10.1006/aama.1998.0633}}.

\bibitem[SSC14]{SSC}
Aaron Snook, Grant Schoenebeck, and Paolo Codenotti.
\newblock Graph {Isomorphism} and the {Lasserre} hierarchy.
\newblock \href{https://arxiv.org/abs/1401.0758}{arXiv:1401.0758 [cs.CC]},
  2014.

\bibitem[TDJ{\etalchar{+}}22]{DBLP:conf/eurocrypt/TangDJPQS22}
Gang Tang, Dung~Hoang Duong, Antoine Joux, Thomas Plantard, Youming Qiao, and
  Willy Susilo.
\newblock Practical post-quantum signature schemes from isomorphism problems of
  trilinear forms.
\newblock In Orr Dunkelman and Stefan Dziembowski, editors, {\em Advances in
  Cryptology - {EUROCRYPT} 2022 - 41st Annual International Conference on the
  Theory and Applications of Cryptographic Techniques, Trondheim, Norway, May
  30 - June 3, 2022, Proceedings, Part {III}}, volume 13277 of {\em Lecture
  Notes in Computer Science}, pages 582--612. Springer, 2022.
\newblock \href {https://doi.org/10.1007/978-3-031-07082-2\_21}
  {\path{doi:10.1007/978-3-031-07082-2\_21}}.

\bibitem[Zui17]{zuiddam}
Jeroen Zuiddam.
\newblock A note on the gap between rank and border rank.
\newblock {\em Linear Algebra Appl.}, 525:33--44, 2017.
\newblock \href {https://doi.org/10.1016/j.laa.2017.03.015}
  {\path{doi:10.1016/j.laa.2017.03.015}}.

\end{thebibliography}
